\documentclass[twoside,leqno]{article}

\usepackage[letterpaper]{geometry}

\usepackage{siamproceedings}

\usepackage[T1]{fontenc}
\usepackage{graphicx}
\usepackage{epstopdf}
\usepackage{enumitem}
\usepackage{algorithmic}
\ifpdf
  \DeclareGraphicsExtensions{.eps,.pdf,.png,.jpg}
\else
  \DeclareGraphicsExtensions{.eps}
\fi

\usepackage{amsmath,amssymb,euscript,enumitem,amsfonts}
\usepackage{tcolorbox}

\usepackage{thmtools} 
\usepackage{thm-restate}

\usepackage{enumitem}
\usepackage{mathrsfs}

\usepackage{cancel}
\usepackage{ulem}

\usepackage{tikz} 
\usetikzlibrary{patterns}
\usetikzlibrary{shapes.geometric, calc}

\usepackage{tcolorbox}
\usepackage{hyperref}
\usepackage{color}
\definecolor{darkgreen}{RGB}{0,100,0}

\definecolor{firebrick}{RGB}{178,34,34}
\definecolor{Andre_blue}{RGB}{0,120,220}
\definecolor{Andre_gray}{RGB}{180,200,200}

\usepackage{abstract}
\renewcommand{\abstractname}{} 

\usepackage{tikz, pgfplots}
\usetikzlibrary{patterns}
\usetikzlibrary{shapes.geometric, calc}

\newsiamremark{remark}{Remark}
\newsiamremark{hypothesis}{Hypothesis}
\crefname{hypothesis}{Hypothesis}{Hypotheses}
\newsiamthm{claim}{Claim}

\definecolor{darkgreen}{RGB}{0,100,0}

\newcommand{\dgm}{\mathrm{Dgm}}
\newcommand\Sp{\mathbb{S}}
\newcommand\R{\mathbb{R}}

\newcommand{\M}{\mathcal{M}}

\newcommand{\Per}{\mathscr{D}}
\newcommand{\PF}{\operatorname{CV}}

\newcommand{\Rspace}       {{\mathbb R}}

\newcommand{\Edist}[2]     {\|{#1}-{#2}\|}

\newcommand{\Bottleneck}   {{W_\infty}}
\newcommand{\Dgm}[2]       {{\rm Dgm}_{#1}{({#2})}}

\newcommand{\Skip}[1]      {}

\pgfplotsset{compat=1.18} 

  {\hfill\hspace*{\fill}~$\square$\end{trivlist}}

\makeatletter
\newcommand\footnoteref[1]{\protected@xdef\@thefnmark{\ref{#1}}\@footnotemark}
\makeatother

\begin{document}

\newcommand\relatedversion{}
\renewcommand\relatedversion{\thanks{The full version of the paper can be accessed at \protect\url{https://arxiv.org/abs/0000.00000}}} %

\title{Braiding Vineyards}
    
\author{Erin W.~Chambers\thanks{Department of Computer Science and Engineering, University of Notre Dame, Notre Dame, IN (\email{echambe2@nd.edu},  \url{https://wolfchambers.github.io/}).}
\and Christopher Fillmore\thanks{Institute of Science and Technology Austria, Klosterneuburg, Austria (\email{cdfillmore@gmail.com},  \url{https://orcid.org/0000-0001-7631-2885}).}    
\and Elizabeth Stephenson\thanks{Orteliu, Oslo, Norway
(\email{elizasteprene@gmail.com}, \url{https://orcid.org/0000-0002-6862-208X}).}
\and Mathijs Wintraecken\thanks{Inria Centre Universit{\'e} C{\^o}te d'Azur, Sophia-Antipolis, France
(\email{mathijs.wintraecken@inria.fr}, \url{https://orcid.org/0000-0002-7472-2220}).}}

\date{}

\maketitle

\begin{abstract} 
In this work, we introduce and study what we believe is an intriguing and, to the best of our knowledge, previously unknown connection between two fundamental areas in computational topology, namely topological data analysis (TDA) and knot theory.  
Given a  function from a topological space to $\mathbb{R}$, TDA provides tools to simplify and study the importance of topological features: in particular, the $l^{th}$-dimensional persistence diagram encodes the topological changes (or $l$-homology) in the sublevel set as the function value increases into a set of points in the plane. Given a continuous one-parameter family of such functions, we can combine the persistence diagrams into an object known as a vineyard, which track the evolution of points in the persistence diagram as the function changes.  If we further restrict that family of functions to be periodic, we identify the two ends of the vineyard, yielding a closed vineyard.  This allows the study of monodromy, which in this context means that following the family of functions for a period permutes the set of points in a non-trivial way.   Recent work has  studied monodromy in the directional persistent homology transform, demonstrating some interesting connections between an input shape and monodromy in the persistent homology transform for 0-dimensional homology embedded in $\mathbb{R}^2$. 

In this work, given a link and a value $l$, we construct a topological space (based on the given link) and periodic family of functions on this space (based on the Euclidean distance function), such that the closed $l$-vineyard contains this link. 
This shows that vineyards are topologically as rich as one could possibly hope, suggesting many future directions of work.  Importantly, it has at least two immediate consequences we explicitly point out:
\begin{itemize}
    \item Monodromy of any periodicity can occur in a $l$-vineyard for any $l$. This answers a variant of a question by Arya and collaborators \cite{Arya2024}. To exhibit this as a consequence of our first main result we also reformulate monodromy in a more geometric way, which may be of interest in itself. 
    \item Topologically distinguishing closed vineyards is likely to be difficult (from a complexity theory as well as a practical perspective) because of the difficulty of knot and link recognition, which have strong connections to many NP-hard problems.
\end{itemize}

\end{abstract}

\section{Introduction}

Computational topology has a number of different branches which, despite some overlap, have remained fairly distinct in nature. On the one hand there is persistent homology~\cite{Frosini1990,Robins1999,Edelsbrunner2002}, a more recent and active area which has its roots in algebraic topology; see~\cite{Oudot2015,Dey2022} for recent survey books on this active topic. On the other hand, there are the computational aspects of knot theory and the study and characterization of low dimensional manifold topology, which has a long history of algorithmic development, perhaps dating back to Dehn's algorithm~\cite{Dehn1911,Dehn1912} and the many following algorithmic results in more recent decades on shape and knot recognition in low dimensions~\cite{Rubinstein92, thompson1994thin, hass1997algorithms, birman1999new, Burton2020}.

The first goal of this paper is to study an interesting new link between these two important branches of computational topology.  Thanks to the work of Alexander~\cite{Alexander1923} we know that every knot or link can be represented as a braid, that is, for every link there is a braid such that if we glue the ends of the braid, we recover the link; see Section~\ref{ssec:knotsandbraids} for more formal definitions. Recall also that for a continuous one parameter family of filtrations, we can ``stack'' the persistence diagrams of these filtrations; we call the resulting object a vineyard \cite{CohenSteiner2006, turner2023representing,Hickok2022b}. Thanks to the stability of persistence diagrams, the points in the persistence diagram move continuously (even Lipschitz continuously) with the parameter. This means that we can follow a point in (the stack of)  the persistence diagrams; the resulting curve is called a vine.
We will prove that for every link there exists an embedded manifold and a family of functions on $\M$ (where each function is induced by the distance to a point in the ambient space, and where in turn each point comes from a curve $\gamma$) such that the vineyard of the family of functions yields the braid representing the link in the sense of Alexander. 

The second goal of this paper is to show that any type of monodromy can occur in vineyards. This is part of a new research direction in computational topology. In~\cite{Arya2024}, the occurrence of monodromy in the context of the directional persistence transform is studied, more precisely for 0-dimensional persistence modules of objects embedded in $\mathbb{\R}^2$. 
Roughly speaking, the directional persistence transform considers the persistence diagram of the height function on a shape for any direction, which in two dimensions gives what we call a closed vineyard (whose precise definition will be given below). 
It should be noted that monodromy had already been identified much earlier in the context of multidimensional persistence in \cite{Cerri2013}, see also \cite{AATRNsara}. 
The authors of~\cite{Arya2024} conclude with an open question about demonstrating monodromy in higher dimensions, as well as several interesting and more open ended questions related to better understanding what monodromy is capturing about the input shape.  

The setting of~\cite{Arya2024} is the directional persistence transform, however this has been placed in a larger context \cite{onus2024shovingtubesshapesgives}, where one studies the distance to flats of any dimension, including points. In this paper, we focus on that extreme case, namely points. The resulting \emph{radial transform} allows for a more geometric understanding of monodromy, and is critical in our construction.

In this paper, we exhibit that any type of monodromy and the braid associated to any link can occur in a vineyard. To make our statement more precise, however, we need to introduce some nomenclature, although full definitions will be deferred until Section~\ref{sec:prelims}. 
Intuitively, monodromy is the effect where if one makes a loop in a base space of a covering or fibre bundle, the lifted curve may not end up at the same point as you started out with. We say that the monodromy is of period $2\pi k$ (with $k>0$) if the lifted curve returns to the starting point after $k$ revolutions in the base space.
In our context the base space is a closed curve or loop $\gamma: [0,2 \pi] \to \mathbb{R}^d$, it is into this image $\mathbb{R}^d$ that we have embedded a manifold $\M$ (link or some offset of link, which is a modification of the input link). The fibres are the persistence diagrams of the Euclidean distance function restricted to the manifold $\M$, that is $d_\mathbb{E} (x, \gamma(t))|_\M$. The bundle therefore is the vineyard. 
The lifted curve is a vine $\tilde{\gamma}_{(b_0,d_0)} (t)$ in the vineyard starting at $(b_0,d_0)$ in the persistence diagram of $d (x, \gamma(0))_\M$ and the periodicity is the smallest $k>0$, such that for all $i$, $\tilde{\gamma}_{(b_0,d_0)} (0)= \tilde{\gamma}_{(b_i,d_i)} (2 \pi k)$, where we assume that the vine is non-degenerate in the sense that it stays away from the diagonal.

More precisely the two main statements are: 
\begin{restatable}[Any link occurs in some closed vineyard]{theorem}{thmLink}
\label{thm:Link}
Given a link {$\mathcal{X}$}, $d,l \in \mathbb{Z}_{\geq 0}$, with $d \geq 3$ and $l<d-2$,  then there exists a submanifold $\M \subset \mathbb{R}^d$ and a closed curve $\gamma \subset \mathbb{R}^d$ such that identifying the ends of the $l$-vineyard of $d (x, \gamma(t))|_\M$ yields a link which contains the given link as a subset. That is, the vineyard is topologically equivalent to the link we were given after removing some spurious connected components, i.e. the output is a link $\mathcal{Z}= \tilde{\mathcal{X}} \sqcup \mathcal{Y} $, where $\sqcup$ denotes the disjoint union of different connected components and $\tilde{\mathcal{X}}$ is ambiently isotopic to $\mathcal{X}$. 
\end{restatable}

\begin{restatable}[Any monodromy occurs in some closed vineyard]{theorem}{thmMonodromy}
\label{lem:ouroboros} \label{thm:Monodromy} 
The radial transform in $\mathbb{R}^d$ can exhibit monodromy for persistence up to the ($d-2$)th homology and for extended persistence up to the ($d-1$)th homology. Moreover the periodicity of the monodromy can be $2k \pi$ for any $k \in \mathbb{Z}_{\geq 2}$.
\end{restatable}

\begin{figure}[h!]
    \centering
    \includegraphics[width=0.95\linewidth]{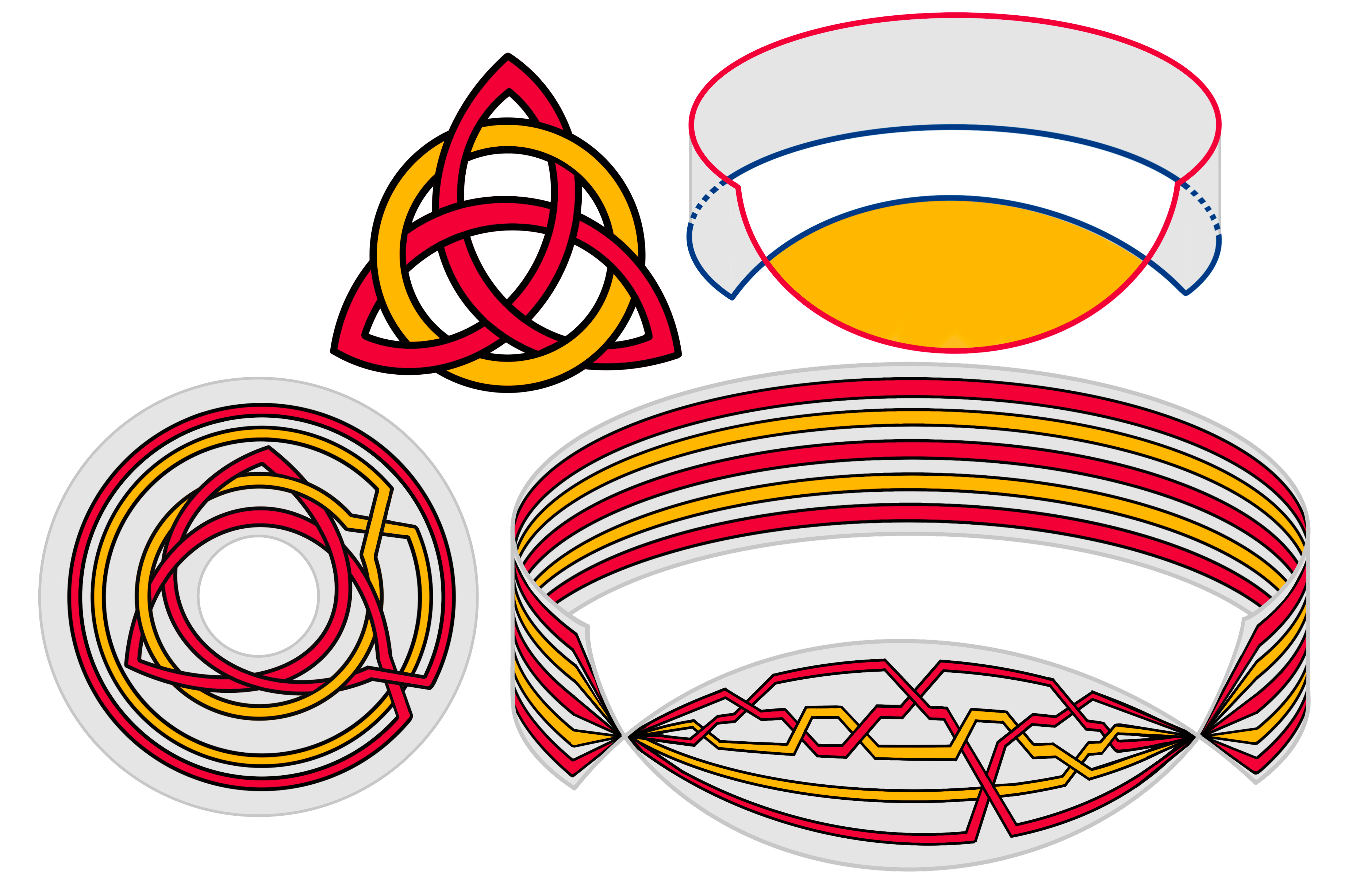}
    \caption{ For a given a link (top left), we construct a closed braid (that is, a braid with ends identified) embedded near an annulus (bottom left). Note that we add an additional trivial loop in each connected braid component traveling around the exterior of the annulus. We then partially twist the annulus (top right) such that the cylindrical vertical portion in gray and the horizontal portion in yellow are perpendicular to each other. In the top right figure, the top edge is shown in red and the bottom edge in dark blue to clarify this twisting. It looks pinched due to the perspective, but is in fact only twisted. The bottom right figure shows the result of twisting the braid together with the annulus, with some massaging to keep the crossings in the horizontal part. Due to the addition of the exterior trivial loops from the bottom left figure, the vineyard consisting of the persistence diagrams of the distance function to a point following the twisted annulus (bottom right) on the `outside' contains the input braid with some surgery.
    See  Figure \ref{fig:computed-braids} for several 3D views of the embedded knot, as well as Figure \ref{fig:pretty_trefcircle} and the corresponding detailed example and discussion in Appendix~\ref{app:Example}.      } 
    \label{fig:mainResult1}
\end{figure}

\begin{figure}
    \centering
    \includegraphics[width=.75\textwidth]{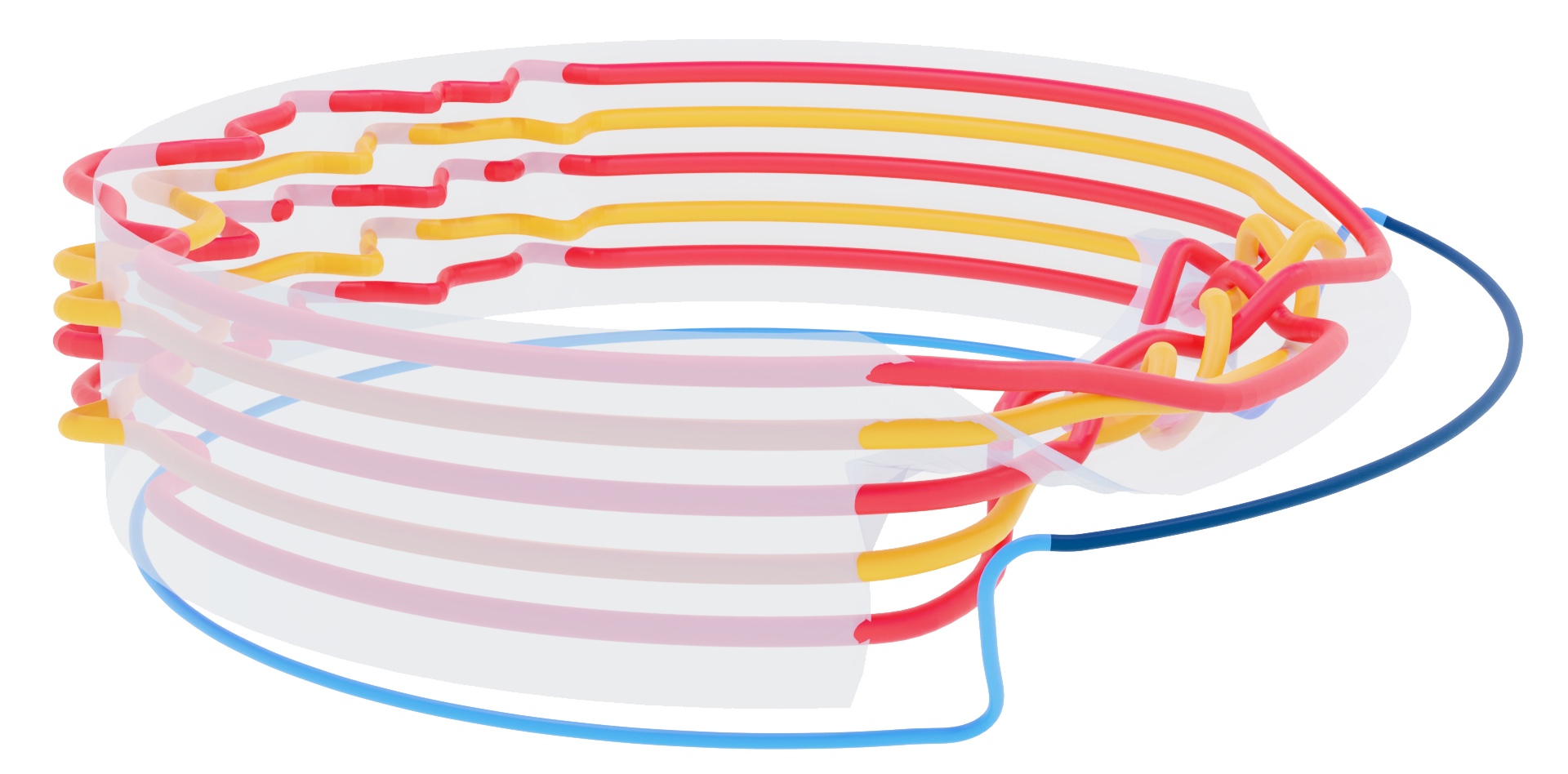}
    \includegraphics[width=.75\textwidth]{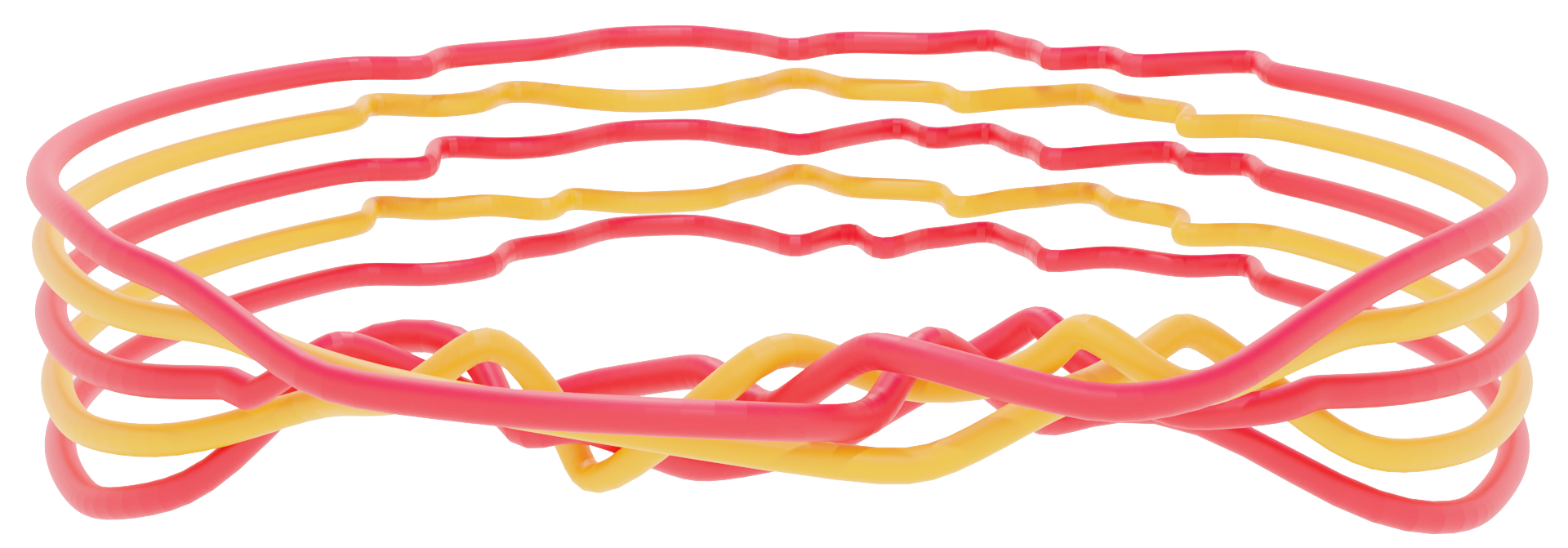}
    \includegraphics[width=.75\linewidth]{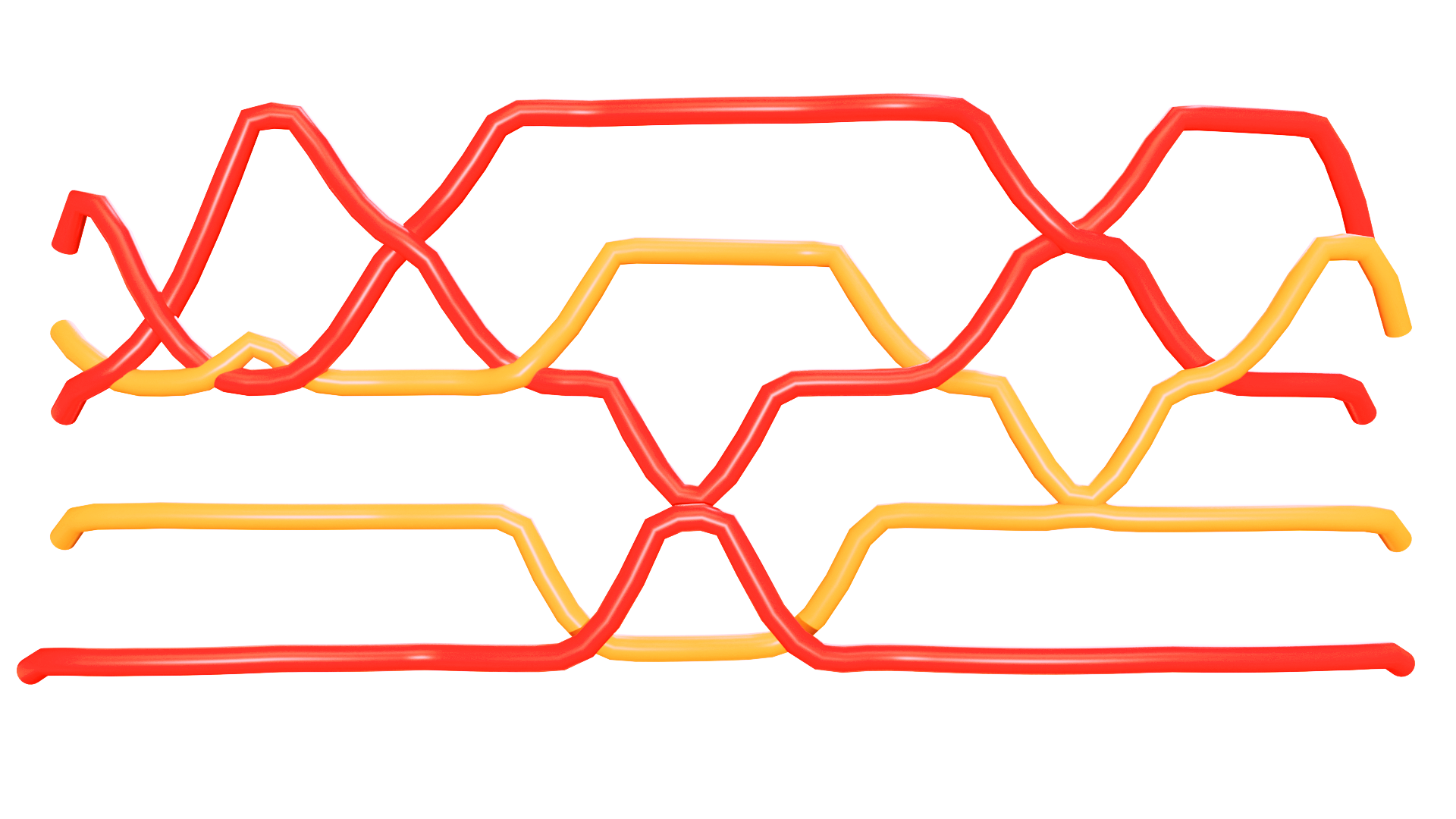}
        \caption{Top: a side view of a 3-dimensional embedding of the link depicted in Figure \ref{fig:mainResult1}, where the crossings and strands are perturbed in a particular way to get the desired monodromy (see the description in Figure \ref{fig:mainResult1}).  The red and gold portions are the two sections of the link, and the blue curve is an observation loop for the radial transform, where the dark blue portion corresponds to the section of vineyard visualized in the bottom. 
        See also the proof of Theorem~\ref{thm:Monodromy} and the corresponding Example in Appendix \ref{app:Example}, for full details on the perturbation. 
        Middle: Front-angled view of the embedded link. 
        Bottom: A sideways view of 
        our computed vineyard, showing the diagrams computed from the radial transform of observation points 
        taken from the fraction of the full period (from $0$ to $2\pi$) depicted in dark blue segment of the 
        curve above.
        This segment of the vineyard captures all crossings and exhibits monodromy of period $2\pi \cdot 3$. To transform this braided vineyard into a closed braid, we would identify the ``sides'', or slices, at $0$ and $2\pi$.}
    \label{fig:computed-braids}
\end{figure}

We remark that our construction hinges not only on $\mathcal{M}$ but also on a careful choice of $\gamma$. Indeed, fixing $\mathcal{M}$, there can be closed curves on $\mathcal{M}$ which yield no links in the vineyard.  

\section{Preliminaries}
\label{sec:prelims}

\subsection{Monodromy} 

Monodromy is an important concept in mathematics that appears in various guises. We refer to the review \cite{ebeling2005monodromy} (the first part of which is almost a review of reviews) and the other reviews mentioned in that paper for an overview of the various aspects of the theory. In this paper we will only consider the simplest incarnation, and only in the setting of topological data analysis. 

Let $\tilde{X}$ be a covering space of $X$ with covering map $C: \tilde{X} \to X$, that is for every $x \in X$ there exists an open neighbourhood $x\in U$ and a discrete set $J$, such that $C^{-1} (U) = \sqcup_{i\in J} V_i$ and $C|_{V_j} : V_j \to U$ is a homeomorphism for all $j\in J$. See Figure \ref{fig:General_covering_space} for an illustration.  We call the inverse images of points $x \in X$ of the map $C$ the fibres. For a curve  $\gamma: [0,2 \pi] \to X$ we write $\tilde{\gamma}$ for (one of) its lift(s), that is a continuous map $\tilde{\gamma} : [0,2 \pi] \to \tilde{X}$ such that $C \cdot \tilde{\gamma} = \gamma$. If $\gamma$ is a loop, that is $\gamma(0)= \gamma(2 \pi)$, then we say that $\gamma$ exhibits monodromy (at the starting point $\tilde{\gamma} (0) \in C^{-1} ( \gamma(0))$) if we have that the start and end points of its lift $\tilde{\gamma}$  are different, i.e. $\tilde{\gamma}(0) \neq \tilde{\gamma}(2 \pi)$. The difference between $\tilde{\gamma}(0)$ and $\tilde{\gamma}(2 \pi)$ is also referred to as monodromy (this difference can in certain cases be best represented by a group, see \cite{ebeling2005monodromy}, although we will not need this in our discussion).

\begin{figure}[h!]
    \centering
    \includegraphics[width=0.5\linewidth]{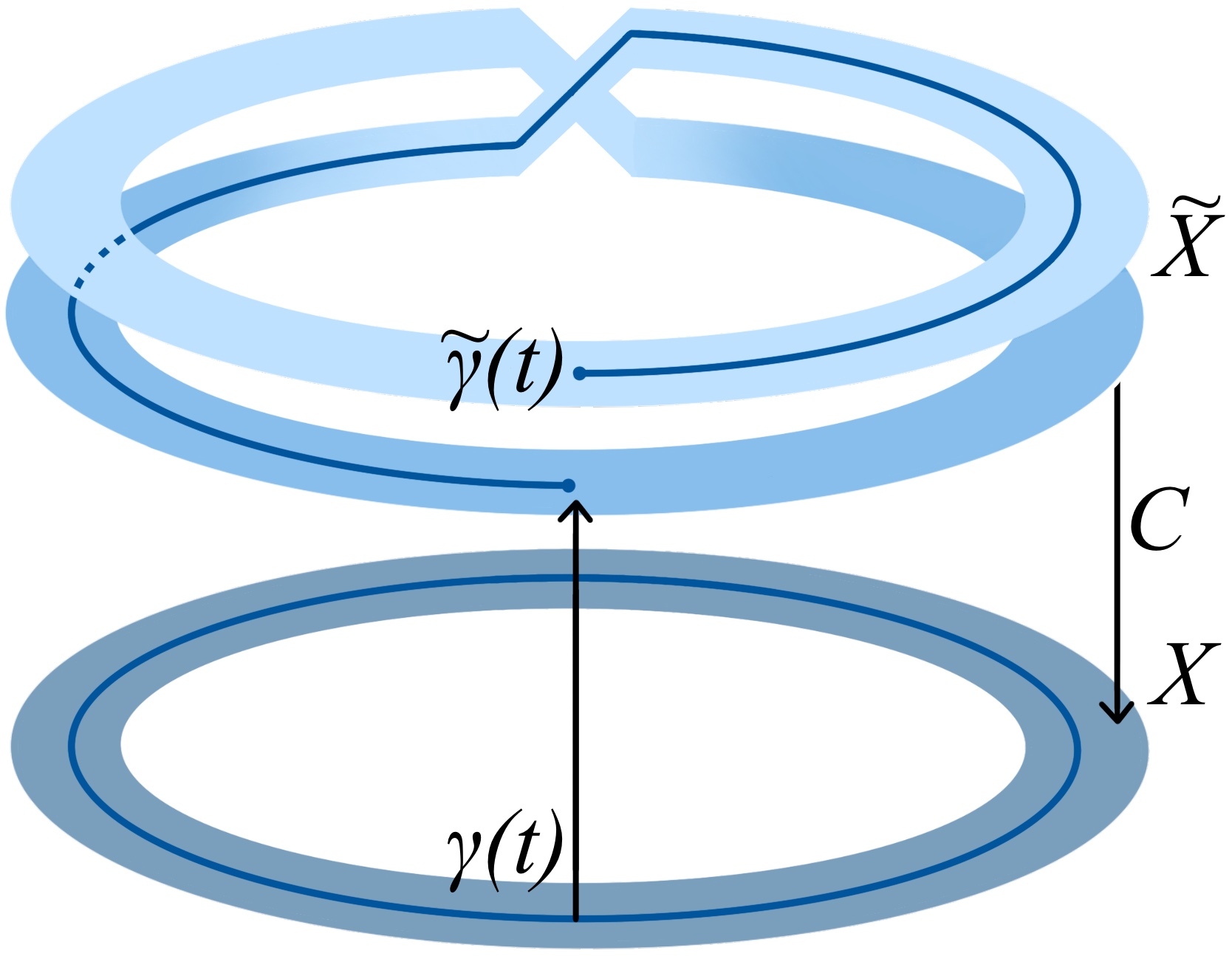}
    \caption{Here we see a cover $\tilde{X}$ (in this case a double cover) of the base space $X$, in this case a circle as well as the curve $\gamma$ and its lift $\tilde{\gamma}$.  %
    }
    \label{fig:General_covering_space}
\end{figure}

If $\gamma$ is a loop we can extend it, formally speaking by concatenating with itself. Here we adopt the convention that if $\gamma$ and $\tilde{\gamma}$ are two curves parametrized by $[0,2\pi]$, then the concatenation $\gamma \circ \tilde{\gamma}$ is parametrized by $[0,4\pi]$, that is, we do not rescale the parametrization interval. We write 
\[ \gamma^k = \underbrace{ \gamma \circ  \dots \circ \gamma}_{k}
\] 
and $\widetilde{\gamma^k}$ %
for its lifting. We stress that generally 
\[ 
\widetilde{\gamma^k} \neq \underbrace{ \tilde{ \gamma} \circ  \dots \circ \tilde{\gamma}}_{k},
\] 
where $\tilde{ \gamma}$ is the lifting of $\gamma$. In fact the right hand side does not even have to be a continuous curve. 
We say that a loop $\gamma$ (parametrized by $[0,2 \pi]$) in the base space $X$ exhibits \emph{monodromy of order} $k$ if $k$ is the smallest positive integer such that the lifted curve $\widetilde{\gamma^k}$ satisfies
\[ 
\widetilde{\gamma^k} (0)= \widetilde{\gamma^k}(2 \pi k). 
\]
If $k=1$ we say that $\gamma$ exhibits no or trivial monodromy.

\subsection{Knots, links, and braids} 
\label{ssec:knotsandbraids}

In this subsection we briefly recall the formal definitions of knots, links and braids. 
An (oriented) \emph{knot} is the the equivalence class of oriented closed curves embedded in $3$-dimensional Euclidean space, $\gamma \colon \Sp^1 \to \R^3$, under ambient isotopy.
    To simplify the discussion, we will abuse notation and write $\gamma$ for the map as well as its image in $\R^3$.
    A \emph{link} with $n$ \emph{components} is a disjoint union of $n$ knots, $L=\gamma_1, \cup \ldots \cup \gamma_n \subset \R^3$.
We say two knots (or links) are \emph{equivalent} if there is an orientation-preserving homeomorphism from $\R^3$ to $\R^3$ such that one knot (or link) is the image of the other.
Generally, this equivalence is formalized as a series of Reidemeister moves (see Figure~\ref{fig:Reidemeister}), as any isotopy between knots can be related by a sequence of these three local moves~\cite{Alexander1926,Reidemeister1927}. 

\begin{figure}
\begin{center} 
\includegraphics[width=0.7\linewidth]{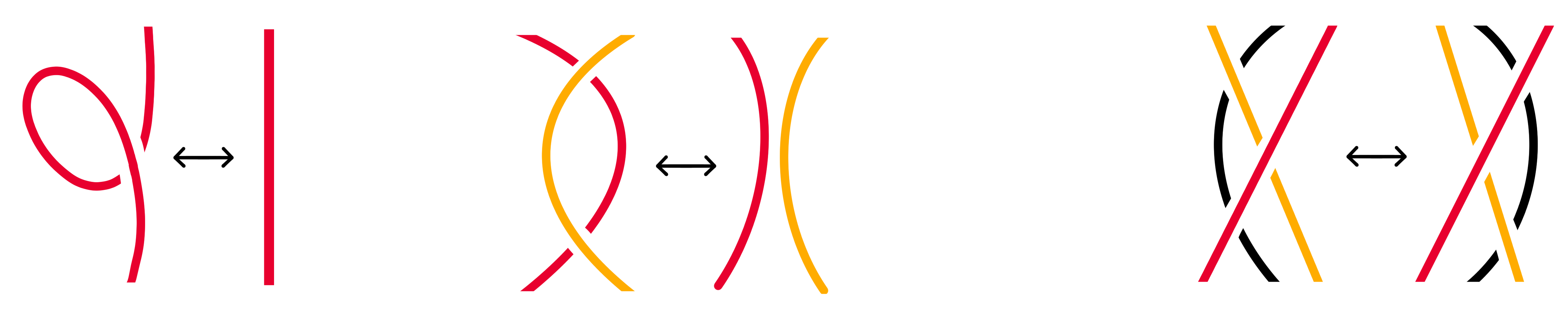}
\end{center} 
\caption{
The Reidemeister moves. Left to right: Type I, Type II, Type III.
\label{fig:Reidemeister}
}
\end{figure}

For each $u \in \Sp^2$, the projection of a knot (or link), $\gamma$, in the direction $u$ provides a \emph{knot diagram} (or link diagram) of $\gamma$ with crossings.  In this paper we will assume diagrams with generic crossings, so that our diagram has a finite number of crossings and no self-tangencies.

A \emph{braid} on $m$ strands is the disjoint union of $m$ intervals embedded in a solid\footnote{That is the product of the disk $D^2$ and the interval $I$.} cylinder, that is the braid $B$ is the union of strands $B_i$, where $B_i\colon I \to D^2\times I$, monotonically increasing with respect to $I$, such that $B_i(0) = (d_i, 0)$ and $B_i(1) = (d_j, 1)$ for each  $i$, where $i$ indexes the strand, so that the set of endpoints of strands is some permutation of the set of origins.
Two braids are \emph{equivalent} if there is an orientation-preserving homeomorphism from $D^2\times I$ to $D^2\times I$ such that one braid is the image of the other and at each time during the homeomorphism the image is also a braid.
For each $u \in \Sp^1$, the projection of a braid in the direction $u$ provides a \emph{braid diagram} of $B$ with crossings.  In this paper we will assume braid diagrams with generic crossings, and we further use that each braid can be represented in a piecewise vertical diagram, that is the strands in the diagram are vertical except in a (small) neighbourhood of a crossing. The fact that braids can always be represented in such a way seems to be folklore; see for example ``Reidemeister’s Theorem" and the discussion in~\cite{Birman2005}.

\begin{definition}[Closed braid] \label{def:ClosedBraid}
    A \emph{closed braid} or \emph{braided link} is the image of a braid under the map from the solid cylinder to the solid torus, $D^2\times I \to D^2 \times \Sp^1 {\subset \mathbb{R}^2 \times \mathbb{C}  \simeq \mathbb{R}^4}$, which sends $(x,y)\mapsto(x,e^{ iy})$, {where $e^{iy}$ gives the standard embedding from $\mathbb{R}/ 2 \pi \mathbb{Z}$ into $\mathbb{C} \simeq \mathbb{R}^2$}. 
    The \emph{braid index} is the minimum number of strands required to form a closed braid equivalent to a given link. 
\end{definition}

Under the standard embedding $\mathcal{T}$ of the solid torus in $3$-dimensional Euclidean space, where the embedding is rotationally symmetric around the $z$-axis and where  each strand is oriented positively, the closed braid can be considered a link with $n \leq m$ components.  Importantly, the orientation of each resulting component is aligned to a positive orientation on the core circle of the solid torus at all points. 

We note that previous work has connected braids and braid groups with monodromy~\cite{CogolludoAgustn2011, Cohen1997, salter2023stratifiedbraidgroupsmonodromy, Salter2024}, further motivating the connection which we explore in this paper.

The following result of Alexander will be essential to our result.
\begin{theorem}[Alexander 1923 \cite{Alexander1923}] \label{thm:Alexander}
    Every knot or link is equivalent to the oriented image of a closed braid.
\end{theorem}
Of course this correspondence is not bijective as each link may be equivalent to many closed braids.  An algorithmic alternative proof of this result was later given in \cite{Vogel1990}.  
The complexity of the algorithm depends on the number of Seifert circles, which are defined as:
\begin{definition}[Seifert Circle]
Given an oriented link diagram, by eliminating each crossing and connecting each incoming strand with its adjacent outgoing strand we obtain a diagram of oriented circles known as \emph{Seifert circles}. %
\end{definition}
The number of elementary operations of the algorithm in \cite{Vogel1990} to obtain a braid diagram from a given link diagram with $n$ crossings and $p$ Seifert circles is at most $(p-1)(p-2)/2$ and the number of crossings in the resulting braid is at most $n + (p-1)(p-2)$. 
The braid index of a link is the smallest number of strands needed for a closed braid representation of the link. The braid index is equal to the minimal number of Seifert circles in any diagram of the braid \cite{yamada1987index}. %

\subsection{Persistence and Vineyards}

We  assume that the reader is familiar with the basic topics in algebraic topology \cite{munkres2018elements, hatcher2002algebraic}, Morse theory and handle decompositions \cite{milnor1963morse}, and the theory of  persistent homology \cite{Dey2022, Oudot2015}, given the tools needed in our construction. We nevertheless give a very brief schematic overview of persistence because we need to contrast this with extended persistence, which is in the opinion of the authors not as utilized as it should be, and therefore may deserve to be recalled. However, we do not aim to be complete in our schematic overview, and refer the curious to \cite{Dey2022, Oudot2015} for further details. 

We fix the following blanket assumption throughout the entirety of the paper:
\textbf{We assume that all our persistence diagrams only contain a finite number of points (counted with multiplicity).}  We note that this is always true for $C^2$ manifolds with a tame Morse function, which is the setting we require.\footnote{Note that we are not referring to tame knots but rather tame Morse functions in this work, as an embedded tame knot may have an infinite number of critical points in its distance or height filtration.}

\subsubsection{Persistence}

Topological data analysis (TDA) starts from the premise that while classical homology is useful for studying topological features of data, it lacks the discernment to pick the correct scale to extract features from. Persistent homology remedies this shortcoming by studying nested sequences of sublevel sets called \emph{filtrations}, thereby extracting features across many scales. 
Taking a manifold $\M$ and a nice function to filter it with, for example, the height function, yields the \emph{persistence module} comprised of homology groups and linear maps induced by the inclusions between sublevel sets
\begin{align}
  \ldots \to H(\M_{a_{i-1}}) \to H(\M_{a_{i}}) \to \ldots \to H(\M_{a_{j-1}}) \to H(\M_{a_j}) \to \ldots .
\end{align}
Composing the maps between consecutive groups, we get a map between any two groups in the module.
We say a homology class $\alpha \in H(\M_{a_i})$ is \emph{born} at $\M_{a_i}$ if it is not in the image of the map from $H(\M_{a_{i-1}})$ to $H(\M_{a_i})$.
If $\alpha$ is born at $\M_{a_i}$, it \emph{dies} entering $\M_{a_j}$ if the image of the map from $H(\M_{a_{i-1}})$ to $H(\M_{a_{j}})$ does not contain the image of $\alpha$, but the image of the map from $H(\M_{a_{i-1}})$ to $H(M_{a_{j-1}})$ does. 
The \emph{persistence} of $\alpha$ is the difference between the function values at its birth and its death.
If the function is a Morse function on a manifold, then precisely one Betti number $\beta_{{\tilde{\ell}}}$ changes when the threshold passes a critical value.
If the order of the corresponding critical point is ${\ell}$, then either a ${\ell}$-dimensional class is born, so $\beta_{\ell}$ increases by one, or a $({\ell}-1)$-dimensional class dies, so $\beta_{{\ell}-1}$ decreases by one.
We use a \emph{persistence diagram} $\Dgm{}{f}$ to encode the  birth-death information of all of the $p$-dimensional homology classes of $M$ arising from the sublevel set filtration induced by the filtering function $f$, where each birth/death pair becomes a point in $\mathbb{R}^2$.  See Figure~\ref{fig:NonExtPErsTorus}.

\begin{figure}
\begin{center} 
\includegraphics[width=0.9\linewidth]{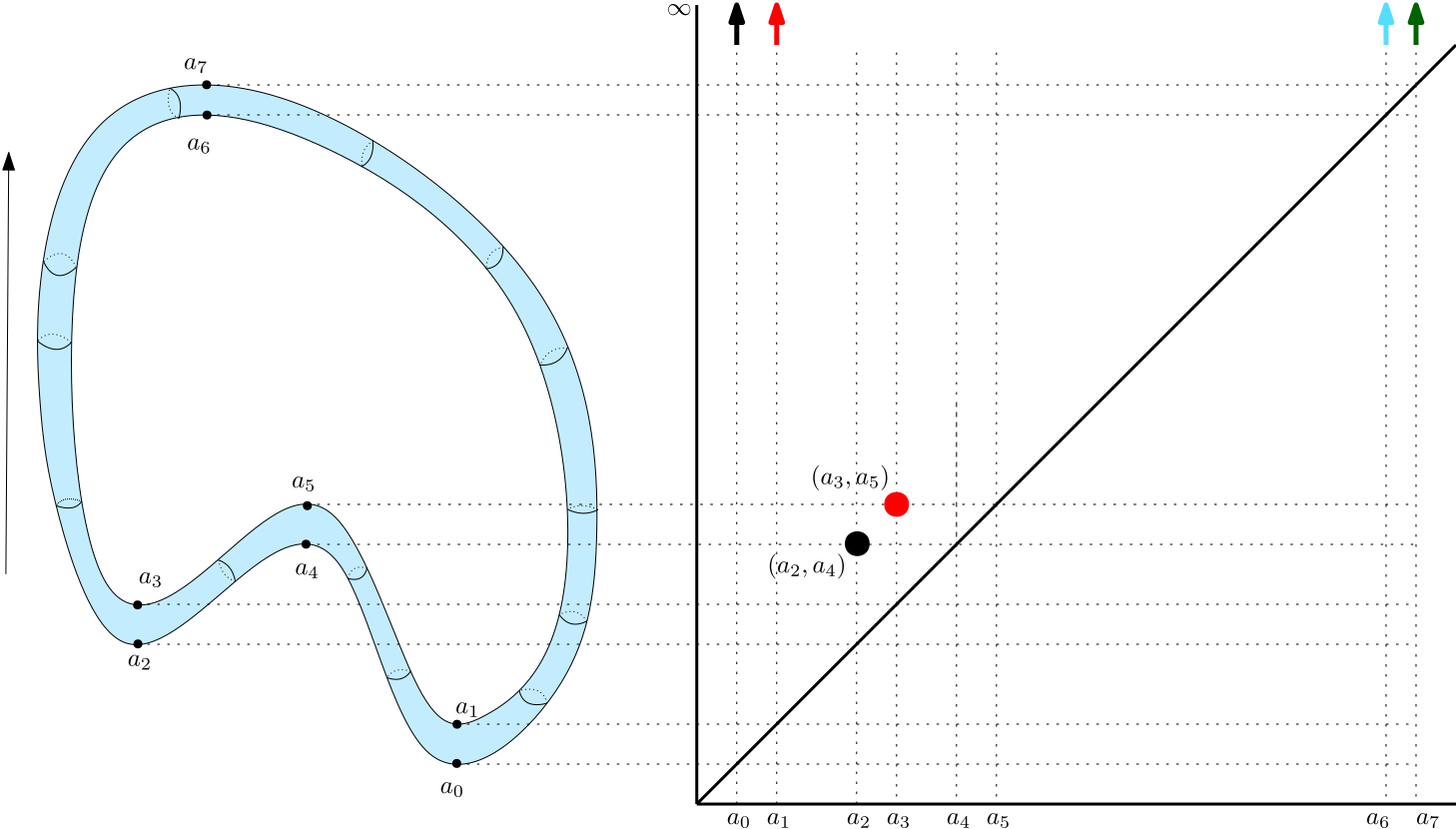}
\end{center} 
\caption{
We consider a torus ($\mathbb{S}^l \times \mathbb{S}^1$) embedded in $\mathbb{R}^{l+2}$  with a height function (whose direction is indicated by the black arrow). 
In the figure, the color indicates the degree in homology of the point: black for $H_0$; light blue for $H_1$; red for $H_l$; and dark green for $H_{l+1}$. Arrow markers distinguish cycles that live forever from those with finite death coordinates. %
\label{fig:NonExtPErsTorus}
}
\end{figure}

\subsubsection{Extended persistence}

There are some drawbacks to the topological summary given by standard persistence, most notably from our perspective the points at infinity, which could cause difficulties in the next section when we formalize the notion of a braided vineyard.  For example, consider a sublevel set filtration of a shape embedded in $\mathbb{R}^2$ or $\mathbb{R}^3$ with nontrivial $H_0$ and $H_1$: at some point, connected components and loops are born, but never die, as they are present in all sublevel sets after their initial appearance.  Even worse, if the input is non-generic and two $H_1$ cycles are born at the same height in the sublevel set filtration, then they give rise to identical persistence pairs.

To address this, Agarwal et.~al.~\cite{Agarwal2006} established a pairing between all critical points of a height function on a 2-manifold, which Cohen-Steiner et.~al.~\cite{CohenSteiner2008} extended to general manifolds with tame functions, leveraging Poincar\'{e} and Lefschetz duality to create a new sequence of homology groups where we begin and end with the trivial group. This guarantees that each homology class that is born will also die at a finite value, replacing all the problematic points paired with $\infty$, and guaranteeing a perfect matching on critical points. 
Let $H_{l}(\M,\M^a)$ denote the relative homology group of $\M$ with the superlevel set $\M^a = f^{-1}[a,\infty)$. 
Again assume we have a tame function %
and critical set $a_1,\cdots,a_k$, and note that $H_{l}(\M_{a_k}) = H_{l}(\M) = H_{l}(\M,\M^a)$ for any $a >a_k$.
From this, we can create a new sequence of homology groups
\begin{align}
    0 & \rightarrow H_{ l}(\M_{a_1}) \rightarrow H_{l}(\M_{a_2}) \rightarrow \ldots \rightarrow H_{ l}(\M_{a_{k-1}}) 
    \nonumber
    \\
    &  \rightarrow H_{l}(\M_{a_k})  = H_{l}(\M,\emptyset) \rightarrow \ H_{l}(\M,\M^{a_k}) 
    \nonumber
    \\
    & \rightarrow H_{ l}(\M,\M^{a_{k-1}}) \rightarrow \ldots \rightarrow H_{l}(\M,\M^{a_1}) = H_{l}(\M,\M) = 0,
    \label{Eq:ExactSequence}
\end{align}
which we call the \emph{$l^{th}$ extended filtration sequence}. We define the sequence $H_{l}(\M_{a_1}) \rightarrow \ldots \rightarrow H_{l}(\M_{a_k})$ as the \emph{upwards} sequence and the sequence 
$H_{l}(\M,\M^{a_k}) \rightarrow \ldots \rightarrow H_{l}(\M,\M^{a_1})$ as the \emph{downwards} sequence.

Note that this sequence fits the structure of a persistence module, so like standard persistence, it has a unique interval decomposition~\cite{CrawleyBoevey2015}. 
The only change is that we interpret the persistence points differently in this setting. 
Specifically, the points in the persistence diagram can be partitioned into three different groups: 1) the classes which are born and die in the upwards sequence, which we denote as ordinary persistence points, 2) the classes which are born and die in the downwards sequence, which we call relative persistence points,  and 3) the classes which are born in the upwards sequence and die in the downwards sequence, which we call the extended persistence points.
Further, we associate the order of birth and death intervals with the value $a_i$ for both $H_{ l}(\M_{a_i})$ in the the upwards sequence and $H_{l}(\M,\M^{a_i})$ in the downwards sequence.
We stress that we use the value $a_i$ as the (birth or death) coordinate in the persistence diagram for both $H_l(\M_{a_i})$ in the the upwards sequence and $H_l(\M,\M^{a_i})$ in the downwards sequence.
Note that this means points in the ordinary diagram are entirely above the diagonal; the relative points lie entirely below the diagonal; while the extended persistence points can be on either side.
See Figure~\ref{fig:torus}.

\begin{figure}
\begin{center} 
\includegraphics[width=0.9\linewidth]{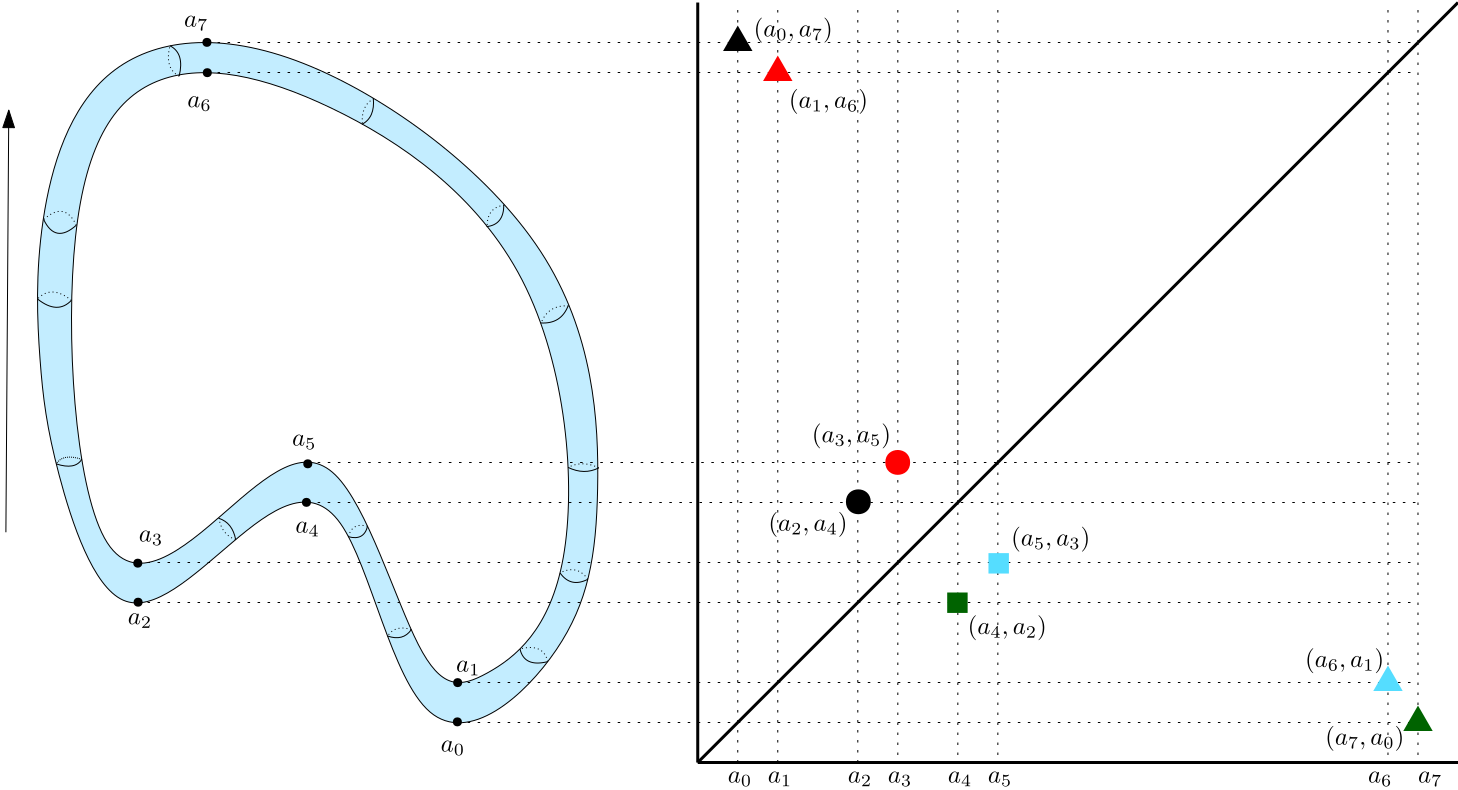}
\end{center} 
\caption{
A torus ($\mathbb{S}^l \times \mathbb{S}^1$) embedded in $\mathbb{R}^{l+2}$  (left) with a height function (whose direction is indicated by the black arrow). 
In %
the figure, the colour indicates the degree in homology of the point: black for $H_0$; light blue for $H_1$; red for $H_l$; and dark green for $H_{l+1}$; while the shape of the marker indicates the type of the point, dots for ordinary points, that is they are both born and die on the way up; squares for relative points, that is they are both born and die on the way down; and triangles for extended points, that is they are born upwards but die downwards.
\label{fig:torus}
}
\end{figure}

\subsubsection{Vineyards}
One can also study the persistence on a manifold $\M$ arising from multiple functions. As long as the functions are similar enough, the Stability Theorem  \cite{cohen2005stability} 
of persistent homology asserts that their associated persistence diagrams will also be similar; see also~\cite{turner2023representing} for a discussion of the algebraic details of  maps between ``nearby" diagrams. More precisely, the \emph{bottleneck distance} {($\Bottleneck$)}  between two persistence diagrams of functions $g_u, g_v \colon \M \to \Rspace$ is bounded from above by the infinity norm of their difference:
\begin{align}
  \Bottleneck (\Dgm{}{g_u}, \Dgm{}{g_v})  &\leq  \Edist{g_u}{g_v}_\infty ,
\end{align}
in which $\Dgm{}{g}$ is the persistence diagram of $g$, a matching between the two diagrams is quantified by the supremum of the max-distances between matched points, and $\Bottleneck$ is the infimum {distance} over all possible matchings, where we allow the introduction of points on the diagonal (where birth time is equal to death time), in which we prefereably match points near to the diagonal to points on the diagonal. Intuitively, the bottleneck distance describes the worst disparity between the best matching of points in persistence diagrams: the worst disparity is the ``bottleneck'' preventing a smaller distance.

This gives rise to the concept of a \emph{vineyard} \cite{CohenSteiner2006,turner2023representing}. It formalizes the idea that a feature of $g_u$ is still recognizable in $g_v$, provided the two functions $u$ and $v$ are not too far apart.
Features are points in the diagram, and the association is a matching between the points of $g_u$ and of $g_v$.  
We now give a brief overview of the definition, but refer the reader to \cite{turner2023representing} for a more detailed description. 
Given a one-parameter family of continuous (Morse) functions (or more generally filtrations) parametrized by $t\in I$ or $t \in \mathbb{S}^1$, with $I$ an interval in $\mathbb{R}$, the vineyard is a section of trivial fibre bundle (consisting of the product of the parameter domain ($I$ or $\mathbb{S}^1$) and the space of persistence diagrams), which associates to each parameter value the persistence diagram of the (sub/super) level set filtration induced by the Morse function. 
In our context (see below) there will be no points in persistence diagrams that have multiplicity higher than one, in which case the structure is simple. Each point on a persistence diagram extends to a Lipschitz curve in the vineyard thanks to the aforementioned stability of persistence diagrams.  In other words each connected component of the section of the bundle is a Lipschitz curve. We call each such connected component a \emph{vine}. We'll come back to this definition in Section \ref{sec:MonodromyInVineyards}, when defining monodromy in our context and introduce additional notation. 
We briefly note that vineyards have been used in several settings for analysis of dynamic data sets~\cite{Yoo2016,Carriere2020,Xian2022,Cipriani2023}.  
Cohen-Steiner et al. \cite{CohenSteiner2006} proposed the so-called \emph{vineyard algorithm} that traces the features (points) while continuously deforming $g_u$ into $g_v$, realized as an update to the reduced matrix arising from $g_u$.
This update can be done in $O(n)$ time, as compared to the alternative (in practice) $O(n^3)$ time\footnote{Formally speaking the complexity is that of matrix multiplication, i.e. $O(n^\omega)$, for which the best currently known value is $\omega\simeq 2.371552$, see \cite{williams2024new}. } of computing and reducing the new matrix from scratch.  The algorithm is available as a part of the package Dionysus~\cite{dionysus}.

\section{Monodromy in vineyards: a geometric viewpoint}
\label{sec:MonodromyInVineyards}

In the context of persistence of some topological space with the induced distance function from a point, we need to introduce a number of conventions for monodromy to be well defined. In fact we will introduce monodromy both associated to an entire vineyard as well as to a single vine in the vineyard. 

\begin{definition}[Radial transform]
We follow the convention from~\cite{onus2024shovingtubesshapesgives} and define the radial transform as the map that associates to any point $p \in \mathbb{R}^d$ the diagram $\Per_l (d_\mathbb{E} (\cdot, p )|_\M) ) \in \dgm $.
\end{definition}

We now consider a  restriction  of the radial transform.
Let $\M$ be a manifold embedded in $\mathbb{R}^d$ and $\gamma: [0,2\pi] \to \mathbb{R}^d$ be a parametrization of a loop $\gamma$. Let 
$d_{\mathbb{E}}(\cdot, \gamma(t) ) |_\M : \M \to \Rspace$ be the function $x \mapsto d_\mathbb{E} (x, \gamma(t))$,
where $d_\mathbb{E} (x, \gamma(t))$ is the Euclidean distance from $x$ to $\gamma(t)$ and, as indicated, $|_\M$ is used to emphasize that we exclusively consider $x \in \M$. %
We will now identify the ends of the interval $[0,2\pi]$ -- that is, we pass to $\Sp^1$. For each $t$ the function $d_\mathbb{E}(\cdot, \gamma(t) )|_\M$  induces a filtration on $\M$, by the sub-level sets of the function. Therefore we can consider (for each $t$) the $l$th order persistence diagram of this filtration $\Per_l (d_\mathbb{E}(\cdot, \gamma(t) )|_\M)$. 
The map 
\begin{align} 
\PF_{\M} : \Sp^1 &\to \Sp^1 \times \dgm \nonumber\\
t &\mapsto (t, \Per_l (d_\mathbb{E} (\cdot, \gamma(t) )|_\M) ), 
\nonumber
\end{align}
where  $\dgm$ is the space of persistence diagrams, is (trivially) a covering space of $\Sp^1$. We refer to $\PF_{\M}$ as the closed vineyard map. 
Thanks to \cite{cohen2005stability} the points in the persistence diagram are (Lipschitz) continuous with respect to $t$. The map  $\PF_{\M}$ is illustrated in Figure \ref{fig:MonodromyPersistence1}.

\begin{figure}[h!]
    \centering
    \includegraphics[width=0.75\linewidth]{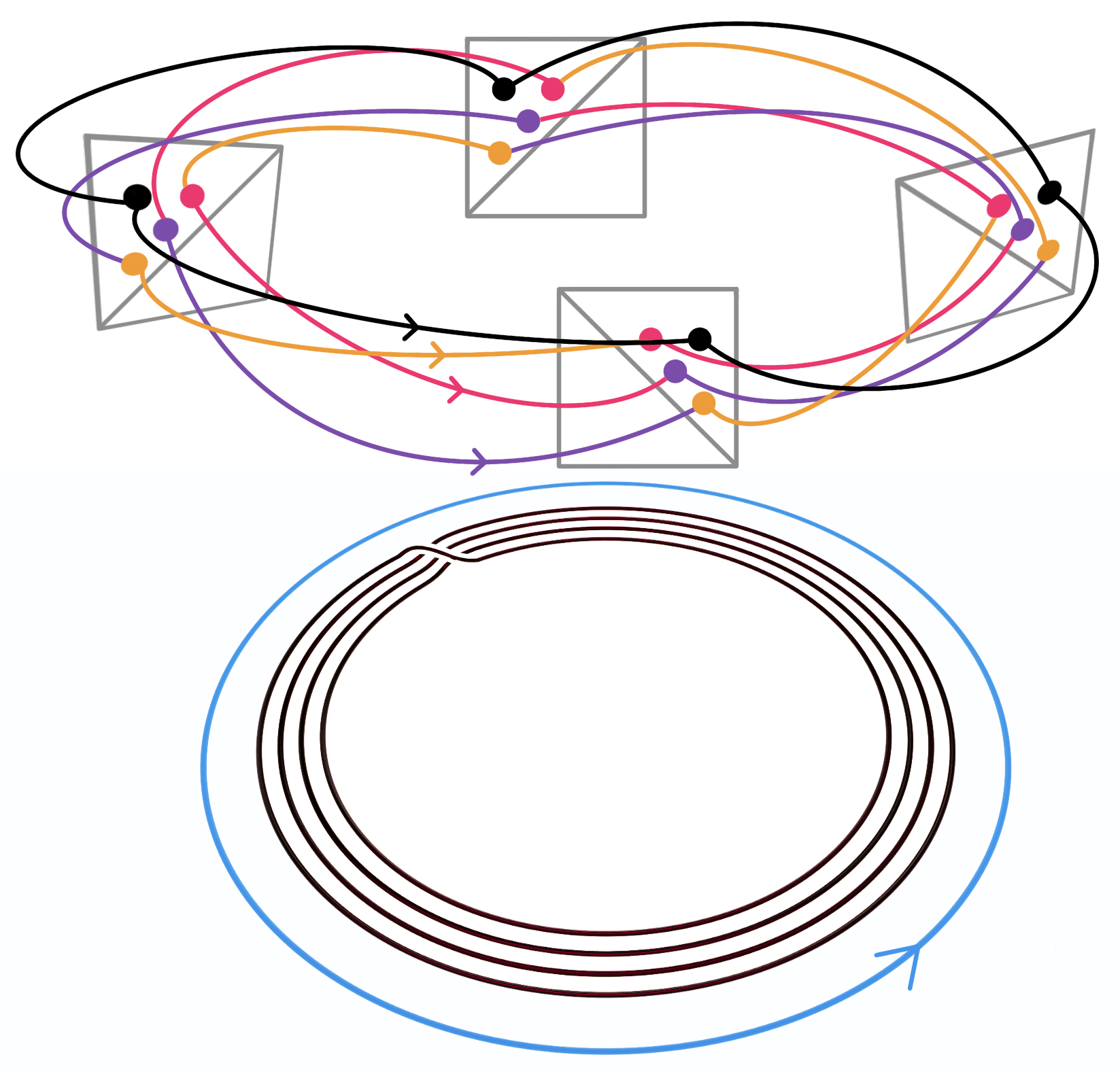}
    \caption{An illustration of the map $\PF_{\M}$. 
    We indicate the persistence diagrams, which form the fibers, only explicitly in a number of places for reasons of visibility. 
    The knot diagram below (black), define a ouroborus knot. Here, we define the ouroborus as an unknot whose knot diagram is a (finite or segment of a) spiral with the end points connected with an monotone curve (where monotonicity refers to the relation between the angular and radial coordinates)  such that the knot diagram only has under or over crossings.  The observation loop is indicated in blue. 
    }
    \label{fig:MonodromyPersistence1}
\end{figure}

We note that generically \cite[page 3]{turner2023representing}, there are no points of higher multiplicity in a persistence diagram and hence the connected components in the image of $\PF_{\M}$, that is the vines, are non-intersecting curves.  We fix the assumption throughout this paper that our vineyards are generic in this sense.

This leads us to the definition of monodromy in the simplest setting (where vines stay clear of the diagonal), which nonetheless is all that is needed for our constructions.  Intuitively, we want to codify monodromy in terms of an integer, which encodes how many copies of the vineyard must be glued together in order for points to return to their start.
More formally:

\begin{definition}\label{Def:Monodromy1}
Assume that for all points $\gamma(t)$, there are no points of multiplicity higher than $1$ in $\Per_l (d_\mathbb{E}(\cdot, \gamma(t) )|_\M)$, which means that all the vines are non-intersecting curves. Assume moreover that all of the vines are disjoint from the diagonal.  We define the following:
\newline 
\textbf{Individual vine monodromy} Given a point in the persistence diagram $\Per_l (d_\mathbb{E}(\cdot, \gamma(\tau) )|_\M)$ for some $\tau \in \mathbb{S}^1$,  write $\mathcal{V} (t)$ for the lift of $\gamma(t)$, which continuously assigns a point in $\Per_l (d_\mathbb{E}(\cdot, \gamma(t) )|_\M)$ for every $t$ and yields the given point for $t =\tau$. We will now assume without loss of generality that $\tau =0$. We say that the vine $\mathcal{V} (t)$ exhibits monodromy if the start and end points of the lifted curve do not coincide. 
Similar to $\mathcal{V}(t)$, write $\tilde{\mathcal{V}}^k (t)$ for the analogous lift of $\gamma^k(t)$. Then $\mathcal{V}(t)$ exhibits monodromy of order $k$ if $k$ is the smallest integer strictly larger than $0$, such that $\tilde{\mathcal{V}}^k (0) = \tilde{\mathcal{V}}^k (2 \pi k)$. 
In other words, the order of the monodromy is the number of points above any $t$ in the connected component in the image of $\PF_{\M}$ that contains the given point in $\Per_l (d_\mathbb{E}(\cdot, \gamma(\tau) )|_\M)$. 
\newline 
\textbf{Vineyard monodromy}
Following an individual vine (with increasing time $t$) for a given point on $\Per_l (d_\mathbb{E}(\cdot, \gamma(0) )|_\M)$ yields a point in $\Per_l (d_\mathbb{E}(\cdot, \gamma(2 \pi) )|_\M) = \Per_l (d_\mathbb{E}(\cdot, \gamma(0) )|_\M)$. In other words the vines or vineyard induce a map $P_\mathcal{V}$ from $\Per_l (d_\mathbb{E} (\cdot, \gamma(0) )|_\M)$ to itself, which permutes the points in the persistence diagram. We say that the vineyard has monodromy of order $k$ if $k$ is the order of the permutation, that is the smallest integer $k>0$ such that applying this permutation $k$ times yields the identity permutation. 
\end{definition}

\begin{figure}[h!]
    \centering
    \includegraphics[width=0.75\linewidth]{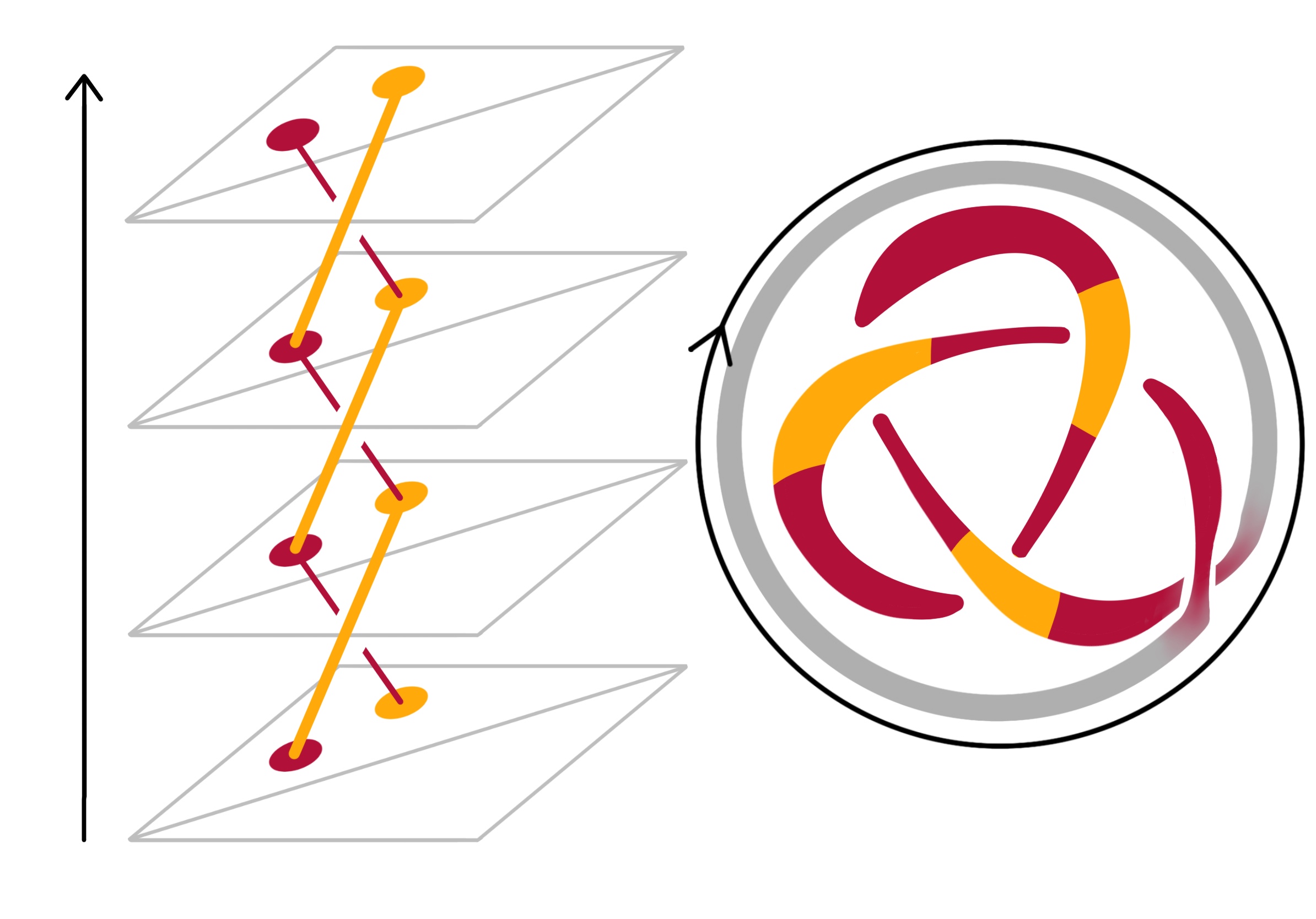}
    \caption{ A schematic overview of how we find monodromy in the vineyard  (left) for a closed braid (right) and observation loop (right, in black). 
    Note that the gray part of the knot does not contribute to monodromy, as the elder rule dictates that the first birth is paired with the last death which causes interference with the desired braiding. By adding a new strand to the outside of the annulus via a trivial twist we ensure that there is an elder vine that sits far away from the diagonal and doesn't interfere with our desired braid.
    To simplify the image, the gray portion  is not shown in the vineyard.  
    }
    \label{fig:Completion1}
\end{figure}

See Figure~\ref{fig:Completion1} for a simple sketch of a knot exhibiting monodromy.  

We conclude here by noting that monodromy can be defined when the vines touch the diagonal, which corresponds more closely to the notion introduced in~\cite{Arya2024}; we refer the interested reader to Appendix~\ref{appendix:diagonal}, where we formalize this definition, but we do not need it for our results.
In addition, there may be self-intersections on the diagonal of the completed vineyard. If the vineyards are non-generic, the situation is significantly more complex as vines can intersect~\cite{turner2023representing}.  As this non-generic case is not needed, we will not consider it any further in this work.

\section{Monodromy and braiding in vineyards}

With the preliminaries and definitions out of the way, we can focus on the main results. The proof of Theorem \ref{thm:Link} is almost completely constructive and proceeds as follows:  
\begin{itemize}
\item We get a link as an input (see Figure \ref{fig:mainResult1} top left).
\item We use Alexander's theorem to represent this as a closed braid that lies very close to an annulus in the `horizontal' plane in $\mathbb{R}^3$, such that the crossings are closely packed. Moreover we add an extra `loop' to the closed braid. This extra loop is there to compensate for the elder rule in persistence  (see Figure \ref{fig:mainResult1} bottom left). 
\item We twist the annulus. This twist is the first step towards independently determining the birth and death times in the persistence diagram (see Figure \ref{fig:mainResult1} right). 
\item We define an observation loop $\gamma(t)$ that follows the twisted annulus on the  `outside', see Figure \ref{fig:NoWiggle}. The vineyard of distance function restricted to the manifold $d_{\mathbb{E}}(\cdot, \gamma(t))|_\M$ at this point has the same topology as the braiding, which in turn is the result of Alexander's theorem with the exception that under and over crossings may be interchanged and the extra loop which we added has been split off as a consequence of the elder rule (compare with Figure \ref{fig:braid_to_ouroboros}). In other words, if we restrict the birth times (of the points in the persistence diagram) depending on the time $t$, then we recover the braiding diagram from Alexander's theorem (deformed), but potentially with the wrong crossings.  
\newline The proof of this statement, i.e. we recover the given closed braid, relies heavily on the fact that the Morse critical points of $d_{\mathbb{E}}(\cdot, \gamma(t))|_\M$ come into two clusters, one close to $\gamma(t)$ and one close to the antipodal point ($\gamma(t+\pi)$) of $\gamma(t)$ on the observation loop, this technical ingredient of the proof is treated in Section \ref{sec:BraidsAndMorse}.
\begin{figure}[h!]
    \centering
    \includegraphics[width=0.75\linewidth]{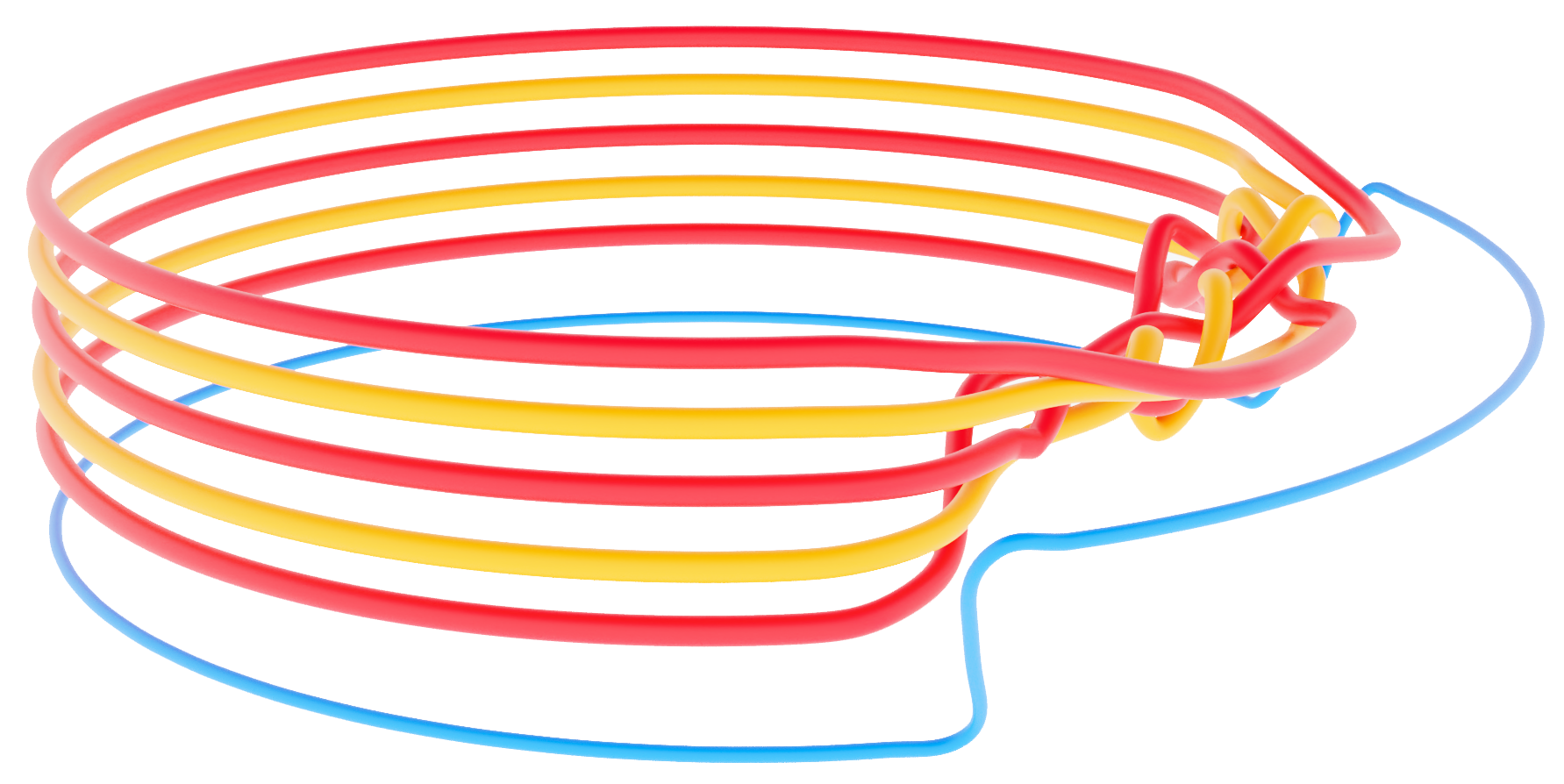}
    \caption{The braid follows a twisted annulus and the observation loop is in its proper place, but without the manipulation of the strands near the antipodal points of the crossing there can be incorrect crossings in the vineyard. 
    }
    \label{fig:NoWiggle}
\end{figure}
\item To ensure that the crossings in the vineyard are the same ones as we were given, 
we need to be able to change the death times of the cycles without influencing the birth times. We now focus on points $\gamma(t)$ near a crossing, as, if we are not near a crossing, the respective heights of the strands are not relevant for the topology. We note that because we consider the Euclidean distance $d_{\mathbb{E}}$ (albeit restricted to the manifold) changing the manifold such that it remains on the same sphere centred on $\gamma(t)$ does not change the birth times. The deaths, on the contrary, occur near the antipodal point ($\gamma(t+\pi)$) on the annulus, and because we have twisted the annulus we can change the birth times by pushing inwards or outwards without influencing the birth times at $\gamma(t+\pi)$, as indicated in Figure \ref{fig:psiManipulation}. Whether we have to push inwards or outwards depends on very precise analysis of the persistence, which takes some four pages in the proof of the main theorem. 
\end{itemize}

\subsection{Braids and Morse theory} \label{sec:BraidsAndMorse}

We begin with some technical preliminary lemmas, which will be essential in our later construction.
We say that a closed braid $B$ is $(\epsilon,R)$-embedded in $\mathbb{R}^3$ if the closed braid is contained in the $\epsilon$-thickening of the circle $C_h(0,R)$ of radius $R$ contained in the horizontal plane. We will  assume throughout this paper that $\epsilon \leq R$. 

We say that the maximal angle that a strand of a braid makes with the horizontal direction is the maximal braid angle $\theta_B$. For a closed braid, we say that the maximal braid angle $\theta_B$ is the angle that the tangent line to a strand at a point $p$ in the braid makes with the normal of the plane $P$ defined by $p$ and the $z$-axis.

We can consider a small $(l+1)$-dimensional $\alpha$-offset $\M$  of a closed braid $B$. 
We interpret `small' here as follows: The offset should be small enough such that no self-intersection, nor intersections with the $z$-axis, occur. In particular, it suffices for the offset to be small compared to the reach~\cite{Federer} of the closed braid. 
Formally, the offset is defined as follows: For $l=0$, we do not add an offset. For $l \geq 1$, we consider the embedding $\mathcal{T} \times 0$ in $\mathbb{R}^3 \times \mathbb{R}^{l-1}$, where $\mathbb{R}^3$ is the space that contains the standard embedding of the torus, in the sense of Definition \ref{def:ClosedBraid}. We take $\M$ to be the offset of the braid $B\times 0$ in $\mathbb{R}^3 \times \mathbb{R}^{l-1}$, that is, the boundary of $(\mathcal{T}\times 0 ) \oplus B(0,\alpha)$, where $\oplus$ denotes the Minkowski sum. The resulting manifold $\M$ can then be embedded in $\mathbb{R}^d$, where the $\mathbb{R}^3$ that contained $\mathcal{T}$ corresponds to the first three coordinates. 
The maximal braid angle in this context is defined as the maximum over the braid angles: Let $p \in \M$ and $T_p\M$ be its tangent space, then the braid angle is the angle between the normal $n$ of the hyperplane $P \times \mathbb{R}^{d-3}$, where $P$ is the plane defined by the $z$-direction and the point $p$. 

\begin{lemma}[clustering of critical points] \label{Lem:Step2} Let $B$ be an $(\epsilon,R)$-embedded closed braid and suppose that $\M$ is its $(l+1)$-dimensional $\alpha$-offset, with braid angle $\theta_B$. If $p \in \mathbb{R}^d$ satisfies $d (p, C_h(0,R)) \leq \eta$, $C_h(0,R)$ is parametrized by $s(\theta)$, and the closest point $p'$ of $p$ on $C_h(0,R)$ satisfies $p' =s(0)$, then the function $x \mapsto d_\mathbb{E}(x, p)|_\M$ has no critical points at the closest point $\beta (\theta)$ on $B$ to $x$ to as long as  
\begin{align}  \frac{\theta}{2} + \theta_B + \arcsin \frac{\epsilon} { 2 R \sin (\frac{\theta}{2})  -\eta } +\arcsin \frac{\eta} { 2 R \sin (\frac{\theta}{2})  }  <\frac{\pi}{ 2} .
\label{eq:BoundL1}
\end{align} 
This implies in particular that there is no topological change as long as \eqref{eq:BoundL1} is satisfied. 
\end{lemma} 
See Figure \ref{fig:Step4} for an illustration of this configuration, and Appendix~\ref{appendix:proofs} for the full proof.

\begin{remark}\label{Rem:ThetaEqZeroOrPi} 
The important conclusion from the bound of Lemma \ref{Lem:Step2} is that by choosing $\epsilon< \eta \simeq \theta_B$ small compared to $\min \{  R, \pi\}$, the bound \eqref{eq:BoundL1} is satisfied as long as $\theta_B \ll  \theta < \pi - 4 \theta_B$. %
\end{remark}

We say that an oriented closed braid $B$ that is $(\epsilon , R)$-embedded is \emph{$\delta$-circular} if its parametrization according to arc length $\beta$ satisfies $\left|\frac{\beta (t)}{R^2} +  \ddot{\beta} (t)  \right| <  \delta$, where we use Newton's notation for the derivative.

\begin{lemma}[Morse indices of clustered critical points] \label{Lem:Step3} 
Let $B$ be an oriented  $(\epsilon , R)$-embedded, $\delta$-circular closed braid, with braid angle $\theta_B$. 
Suppose $p \in \mathbb{R}^d$ satisfies $d (p, C_h(0,R)) \leq \eta$, $C_h(0,R)$ is parametrized by $s(\theta)$, and the closest point $p'$ of $p$ on $C_h(0,R)$ satisfies $p' =s(0)$.

If an oriented  $(\epsilon , R)$-embedded, $\delta$-circular closed braid with $n$ strands, with braid angle $\theta_B$,
$\epsilon \ll R $, and  
\[ 
\frac{6\epsilon}{R} + \delta (R+\eta) \leq \frac{1}{R^2} ( R- \epsilon) ( R- \eta) ,
\] 
and $\epsilon< \eta \simeq \theta_B$ are small compared to $\min \{  R, \pi\}$, 
then $d(\cdot ,p)|_B$ has $2n$ critical points, $n$ maxima and $n$ minima.  
Let $l\geq 1$. If $\M$ is an $(l+1)$-dimensional $\alpha$ offset $\M$ of the same type of braid satisfying the same conditions then $d_\mathbb{E}(\cdot ,p)|_\M$ has $4n$ critical points, $n$ maxima and $n$ minima and $n$ saddle points of Morse index $l$ {and $n$ saddle points of Morse index $1$}.  
\end{lemma}

Again, we refer to Appendix~\ref{appendix:proofs} for the proof.

Given these $4n$ critical points, we can now consider the persistence diagram that results.  Although we work with extended persistence throughout this paper in order to avoid points at infinity and establish a perfect pairing of critical points, we note that in fact for the purposes of establishing monodromy in our knot offset, it suffices to restrict our attention to the behavior of the ordinary points as well as one single extended point, all above the diagonal, as those that contribute to the construction of knot or link as mentioned in Theorem \ref{thm:Link} are born and die `on the way up',  that is in the first line of the sequence in \eqref{Eq:ExactSequence}.

\begin{corollary}[The behaviour of the persistence diagram due to clustering] \label{Lem:Step4} 

Under the same assumptions as in Lemma \ref{Lem:Step3},     
the maxima and minima of $d_\mathbb{E}(\cdot ,p)|_B$ and $d_\mathbb{E}(\cdot ,p)|_\M$  as well as the saddle points of $d_\mathbb{E}(\cdot ,p)|_\M$ can be divided into two groups, one group occurring at low 
values and corresponding to $H_0$ births of cycles in the persistence diagram and one group occurring at high 
values and corresponding to $H_0$ deaths in the persistence diagram.  
Here, under the assumption\footnote{\label{footnote:lowhigh} More generally, Equation \eqref{eq:BoundL1} of Lemma \ref{Lem:Step2} divides $\M$ ($B$ respectively) into three regions, the part close to $p$ where Morse critical points can occur, a large region where no Morse critical points can be found, and finally the part furthest from $p$ where again one may find Morse critical points. The high and low in this corollary should be interpreted as such. } of Remark \ref{Rem:ThetaEqZeroOrPi}, low means $\lesssim  R \,\theta_B $, where the $\lesssim$ hides a constant, and high means $\geq 2R - 6 R \, \theta_B$.
For $H_l$ the situation is identical except for  one change, namely that a single birth occurs at a high value. In `ordinary' persistence theory this cycle lives forever, while in extended persistence it dies at the global minimum, and in fact lies below the diagonal. 
\end{corollary} 

\begin{figure}[h!]
    \centering
    \includegraphics[width=0.95\linewidth]{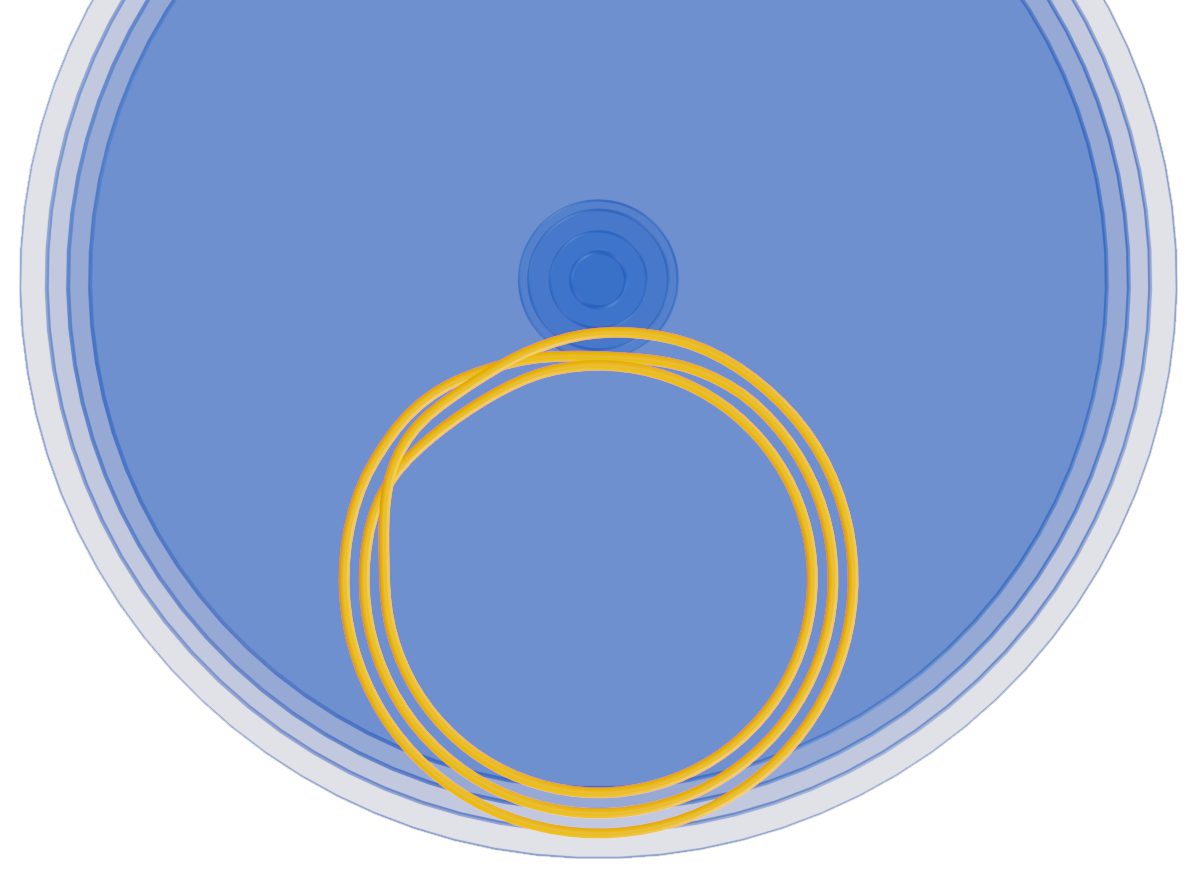}
    \caption{A figure illustrating the statement of Corollary \ref{Lem:Step4}, highlighting the intersection of the ouroboros (in yellow) and a family of 3D growing spheres (in blue) that highlight the $n$ births and $n$ deaths in $H_0$. 
    }
    \label{fig:Step4}
\end{figure}

Let $f_1$ and $f_2$ be two functions on the same manifold $\M$. We say that these two functions are handle-equivalent if the handle decomposition is the same and the times of insertion of these handles are also identical. 
We have the following observation. 
\begin{corollary}[Visualizing as a height function] \label{Cor:Equivalent}
Let $B$ be $(\epsilon,R)$-embedded closed braid and suppose that $\M$ is its $(l+1)$-dimensional $\alpha$-offset, with braid angle $\theta_B$. We have that the $B$ is a circle and the manifold $\M$ is diffeomorphic to the torus $\mathbb{S}^{l+1} \times \mathbb{S}^1$. 
Under the same conditions as Lemma~\ref{Lem:Step3}  (and Corollary~\ref{Lem:Step4}), we have that $d(\cdot ,p)|_B$ and $d(\cdot ,p)|_{\M}$ are handle equivalent to the height function of the embedding depicted in Figure \ref{fig:EquivalentEmbedding}.  
\end{corollary}

\begin{figure}[h!]
    \centering
    \includegraphics[width=0.50\linewidth]{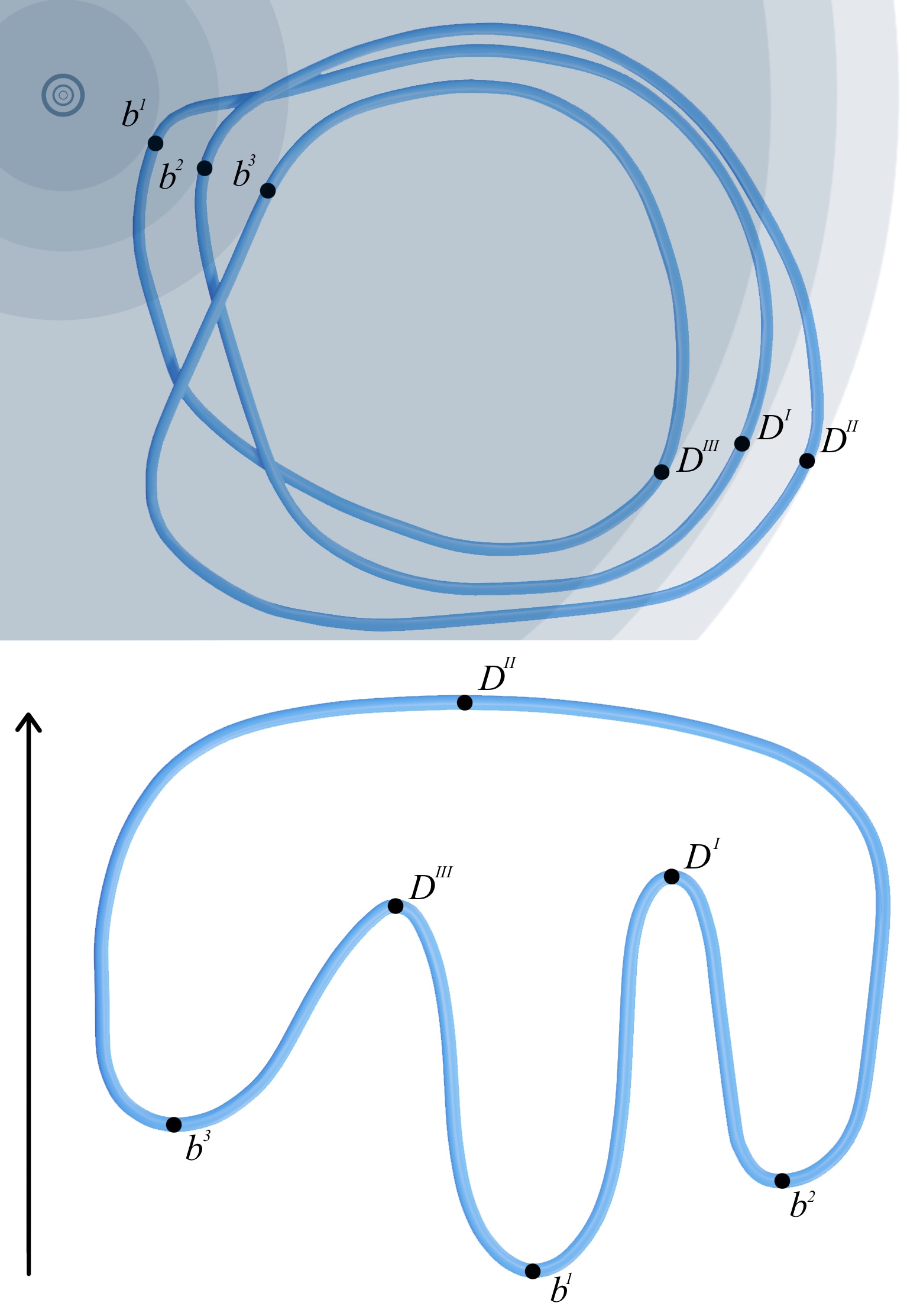}
    \caption{Top: The closed braid. Bottom: The embedding with equivalent height function. The arrow indicates the direction of the height function. The $b^i$s indicate the birth times and $D^J$, where $J$ is a roman numeral {(i.e. $J \in \{ I, II, III\}$)}, the death times. We stress that the picture should be interpreted in a 3D way, and in particular $D^{II}$ does not have to be larger than $D^J$, with $J \neq II$. 
    {We stress that the critical point with the highest value of the Morse function ($D^{II}$ in the figure) corresponds to a death only in extended persistence, in non-extended persistence, only a $1$-cycle is born there.} }
    \label{fig:EquivalentEmbedding}
\end{figure}

\begin{remark}\label{rem:ConditionsAreEasyToSatisfy}
We stress that, given a closed braid $B$, it is not difficult to adjust the embedding such that it is  $(\epsilon , R)$-embedded and {$\delta$-circular} with $\epsilon$ and $\delta$ as small and $R$ as large as you like. Because if $B$ is a closed braid, its parametrization $\beta(t)$ can be written as $\beta (\theta ) = n(\theta) 
+ \rho (\sin (\theta), \cos (\theta) , 0 , \dots, 0)$ for some $\rho>0$, with $n$ normal to  $(\cos (\theta), -\sin (\theta) , 0 , \dots, 0)$. By redefining $\beta (\theta ) = \tilde{\epsilon}  n(\theta) 
+ R (\sin (\theta), \cos (\theta) , 0 , \dots, 0)$ for sufficiently small $\tilde{\epsilon}$ the  {$\delta$-circular} $(\epsilon , R)$-embedding can be achieved. 
By the same argument the observation loop (to be defined in the proof of Theorem \ref{thm:Link}) can be made arbitrarily close in terms of $\eta$ to a circle with radius $R$, so that it has reach as close to $R$ as one likes, by \cite[Theorem 4.19]{Federer}, and hence has a nice tubular neighbourhood of that size.    
\end{remark}

\subsection{Vineyard braiding} %

We will now prove (one of) the main statement(s) of the paper: 

\thmLink*

\begin{proof} %
The proof of this theorem is constructive. Thanks to Alexander's theorem, Theorem \ref{thm:Alexander}, a given link can be represented as a closed braid. {To every connected component of the closed braid we add an extra loop (using a Reidemeister move of type I, so as not to disturb the equivalence class of the closed braid) as in Figure \ref{fig:mainResult1}.} This Reidemeister move may introduce extra crossings, but at most $\mathcal{O} (s\cdot K)$, where $s$ denotes the number of strands in the original closed braid and $K$ the number of connected components of the loop. We now write $C$ for the number of crossings of the resulting closed braid. We write $n$ for the resulting number of strands, for which we see that $n = s+K$. 
We define/construct a particular embedding of the closed braid such that the original braid appears in the vineyard. 
Essentially, the closed braid diagram is parametrized by $\rho$ and $\theta$, while the embedding of the closed braid in $\mathbb{R}^3$ is parametrized by $\rho$, $\theta$, and $h$, as in Figure \ref{fig:braid_to_ouroboros}.
Roughly speaking, the parameter $\rho$ corresponds to the birth time, $\theta$ is the parameter of the observation loop $\gamma$ (to be defined in Step 3) and $h$ corresponds to the death time, again see Figure \ref{fig:braid_to_ouroboros}.

Our construction proceeds stepwise: 
\begin{itemize} 
\item \textbf{STEP 1} We start with an embedding which is close to its annular braid  diagram, by which we mean that the embedding of the closed braid lies in a neighbourhood of an annulus in the plane and the braid is planar with the exception of small neighbourhoods of the points of crossing. We write $\theta, r$ for the coordinates of the annulus, which are polar coordinates in the plane restricted to the annulus.
\item \textbf{STEP 2} We then modify (if necessary) the braid such that the crossings are equally parsed on one side of the annulus, that is, if $\theta$ is the angle that parametrizes the annulus, see Figure \ref{fig:coordinates}, then the crossings are contained in the interval $[0+\frac{\pi}{8(C+1)} , \pi-\frac{\pi}{8(C+1)}]$, where $C$ is the number of crossings. 
By equally parsed we mean that there is only one crossing in each of the intervals $[ \pi \frac{\tilde{j} }{8 (C+1)} , \pi \frac{\tilde{j}+1}{8 (C+1)}  ]$ where $\tilde{j} \in \{  8, 16, \dots , 8 C\}$. 
\item \textbf{STEP 3a} We now modify the embedding of the braid in an angular interval $[\pi-\frac{\pi}{4 (C+1)} , 2 \pi ] \cup [0, \frac{\pi}{2(C+1)}]$. This interval should be interpreted in a periodic manner. 
We do so by twisting the annulus (and by extension the almost annular braid) $90$ degrees in the direction orthogonal to the plane into which the annulus was originally embedded, see Figure \ref{fig:mainResult1}. We do so in such a way that the twisted annulus, and by extension part of the braid in the angular interval $[\pi+\frac{\pi}{8 (C+1)}, 2 \pi -\frac{\pi}{8 (C+1)}]$,  is now close to a cylinder. We do this in a way that preserves cylindrical coordinates, that is, if $\theta$ were the planar angular coordinate of a point on the annulus, then after twisting, $\theta$ is the cylindrical coordinate of the corresponding point. 
We denote the resulting twisted annulus by $\mathcal{A}_T$. 
\item \textbf{STEP 3b} We define an observation loop $\gamma$ to be the curve that follows that twisted annulus on the outside at a constant distance (less than $\eta$, with $\eta$ as in Lemmas \ref{Lem:Step2} and \ref{Lem:Step3}), see e.g. Figure \ref{fig:computed-braids}. %
We also define coordinates $\tilde{\theta}, \rho, \tilde{\psi}$ with respect to this {observation loop}  in a tubular neighbourhood of size $\mathcal{O}(\eta)$ 
of the observation loop. 
This is possible thanks to Remark \ref{rem:ConditionsAreEasyToSatisfy}. The coordinate $\rho$ of a point $y$ in the tubular neighbourhood is the distance to $\gamma$. The coordinate  
 $\tilde{\psi}$ is the angle between a point $y$ in the tubular neighbourhood of the curve, its closest point on curve $\pi_{\gamma} (y)$, and the closest point of this point on the annulus $\pi_{\mathcal{A}_T}( \pi_{\gamma} (y))$. See Figure \ref{fig:coordinates}.

\begin{figure}[h!]
    \centering
    \includegraphics[width=0.55\linewidth]{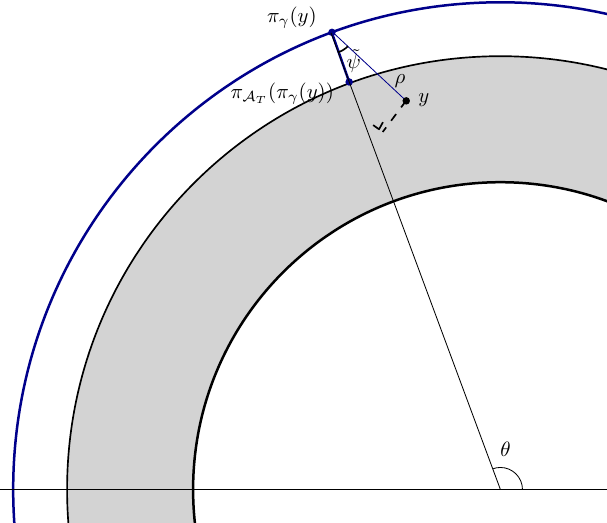}
    \caption{Figure illustrating the notation (the coordinates $(\theta, \rho, \tilde{\psi})$) in step 3b of Theorem~\ref{thm:Link}.  The observation loop is the outer loop shown in dark blue, and (the flat part of) the twisted annulus is shown in grey. 
    }
    \label{fig:coordinates}
\end{figure}
\end{itemize}

\subparagraph{Intermezzo: The persistence diagram for a point on the observation curve} Before we continue with the construction, we now discuss the persistence diagram for a given point. In the following step we'll further modify the closed braid in the angular interval $[\pi-\frac{\pi}{4 (C+1)} , 2 \pi ] \cup [0, \frac{\pi}{2(C+1)}]$, but for now we consider the braid fixed. We also assume for now that the link has only one connected component, that is, it is a knot. 
Thanks to Lemmas \ref{Lem:Step2} and Corollaries \ref{Lem:Step4}-\ref{Cor:Equivalent} we know that the braid and observation curve are chosen in such a way that all the births occur first and then all the deaths occur, in both distance order as well as in order along the braid. Let us in particular consider the equivalent embedding of Corollary \ref{Cor:Equivalent}, where we let $b^1, b^2, \dots b^n $ be the births as they occur in order following along the braid $B$, as in Figure \ref{fig:EquivalentEmbedding}; note that we are slightly abusing notation here, as we are identifying the births with the Morse critical points. Next, let $b^1_{0}, b^2_{0}, \dots, b^n_{0} $ respectively $b^1_{l} , b^2_{l}, \dots, b^n_{l}$, where the subscript indicates if the birth occurs in $0$ or $l$-homology, {be the corresponding ordered births of $\M$}.
Note that we always have $n$ births unless $l=1$, when instead there are $n+1$ births. Again we emphasize that the births are ordered as they occur following $B$ starting with the first birth, not in order of birth time according to the distance filtration;
See Figure \ref{fig:EquivalentEmbedding}. Similarly, let $D^{I}, D^{II}, \dots, D^{N}$  be the deaths as they occur in order if we consider $B$, and $D^{I}_0 , D^{II}_0, \dots, D^{N}_0$ respectively $ D^{I}_l , D^{II}_l, \dots, D^{N}_l$ be the deaths if we consider them on the offset of $B$, $\M$. 
Again we assume that the deaths are ordered as they occur following $B$, not in order of deaths time. 

We will assume without loss of generality that $b^1$ is the lowest birth value. Because of the elder rule, the cycle created at $b^1$ dies (in extended persistence) at the maximal death value, that is $\max_J D^J$, for $B$. Similarly $b^1_{0}, b^1_{l}$ respectively die at $\max_J D^J_0$,  $\max_J D^J_l$ respectively for $\M$. 

The other persistence points are less straightforward, as they  follow the mergers of the sublevel sets of the equivalent embedding of Corollary \ref{Cor:Equivalent}; see Figure \ref{fig:EquivalentEmbedding}. However, the most important case for us will be the following: 
Consider first any ordinary $H_0$ points, so we have $b^1, b^j,b^k$ with $1 \neq j < k \neq 1$. Assume that $b^1$ is the earliest birth. Suppose that the death times are all larger than all birth times 
and are ordered as follows
\begin{align} 
\max \{D^{I}, D^{II}, \dots, D^{J-2}, D^{K+1}, D^{K+2} , \dots D^N\}  &< \min \{ D^{J-1},  D^{K}\} 
\nonumber
\\ 
 \max \{ D^{J-1} ,D^{K}\} &< \min\{ D^{J}, \dots  , D^{K-1}\} ,
\label{eq:DeathConditions}
\end{align} 
then the connected component born at time $b^j$ merges with the connected component born at time $b^1$ at time $D^{J-1}$ and the connected component born at time $b^k$ merges with the connected component born at time $b^1$ at time $D^{K}$ for $B$. 
In other words, under these conditions, there are points $(b^j,D^{J-1})$ and $(b^k , D^K)$ in the persistence diagram. 

{We now consider the persistence diagram for the $l$-homology %
of the manifold $\M$.  
In general, as in Figure \ref{fig:torus}, the births in $l$-homology follow the births in $0$-homology closely. By this we mean the following: As before we write $B$ for the braid and $\M$ for its $(l+1)$-offset, and we compare the same Morse function (distance to a point, or height if we consider the equivalent embedding) on these two spaces ($B$ and $\M$, respectively). 
The birth of $0$-cycles ($b^1,b^2,b^3$ in Figure \ref{fig:EquivalentEmbedding}) on the braid $B$ are close both geometrically and with regard to the value of the Morse function (the distance to a point or the height function, where in the latter case we consider the equivalent embedding) to the critical points that give rise to the birth of $0$-cycles and $l$-cycles on $\M$.  
There is also a simple correspondence with one exception between deaths of $0$-cycles for the braid and deaths of $0$- and $l$-cycles on $\M$, meaning that the critical points that correspond to deaths for $B$ are close to a pair of critical points (one maximum and one saddle) that correspond to deaths in $0$- and $l$-homology.  
The exception is the last death of a $0$-cycle in extended persistence on $B$; here we instead have a Morse critical point which corresponds to the birth of a $1$-cycle for the braid and which lives forever in the non-extended persistence, but corresponds to a death of a $0$-cycle in extended persistence. There are again two corresponding nearby critical points on $\M$ for this final critical point on $B$, however in this case the saddle corresponds to the birth of a $1$-cycle (which corresponds to $\mathbb{S}^1$ in $\mathbb{S}^1 \times \mathbb{S}^l$). This means that there is an extra point in the $l$-persistence diagram if $l=1$. Most importantly this point can be distinguished by the fact that its birth time is much higher than all other points in the persistence diagram.  
However, because the births and deaths of $B$ and $\M$ are so intimately linked except for the final death, we have the following: 
Consider $b^1_l, b^j_l,b^k_l$ and  $1 \neq j < k \neq 1$. Assume that $b^1$ is the earliest birth. If $l =1$ further assume that $b^m_{l=1}$ is the final birth and $j \leq m \leq k-1 $.
If now moreover the death times are ordered as follows
\begin{align} 
\max \{D^{I}_l, D^{II}_l, \dots, D^{J-2}_l, D^{K+1}_l, D^{K+2}_l , \dots D^N_l\}  &< \min \{ D^{J-1}_l,  D^{K}_l\} 
\nonumber
\\ 
 \max \{ D^{J-1}_l ,D^{K}_l \} &< \min\{ D^{J}_l, \dots  , D^{K-1}_l\} ,
\label{eq:DeathConditionsHl}
\end{align} 
then the $l$-cycle born at time $b^j_l$ merges with the $l$-cycle  born at time $b^1_l$ at time $D^{J-1}_l$ and the $l$-cycle born at time $b^k_l$ merges with the $l$-cycle  born at time $b^1_l$ at time $D^{K}_l$ for $\M$. 
In other words, under these conditions, there are points $(b^j_l,D^{J-1}_l)$ and $(b^k_l , D^K_l)$ in the persistence diagram. 
}

\begin{itemize} 
\item \textbf{STEP 4}
As discussed in the intermezzo the first birth is always coupled to the last death in the persistence diagram. We write $b^{j,c}_k(t)$ and $D^{J,c}_k(t)$ for the births, deaths respectively of $d( \cdot , \gamma (t) )|_{\M}$ per connected component $c$, where we stick to the convention that $b^{1,c}_k(t)$ is the first birth (in $k$-homology) for each connected component $c$. We use similar notation for $B$. Here, we note that although this distance value is continuous, the Morse critical point where this minimum is attained at is not continuous.  
A similar effect was called a Faustian interchange in \cite{ElizabethsThesis}. 
With this notation we can conclude that our first observation implies that for each $c$ there is a point $( b^{1,c}_k(t), \max_J  D^{J,c}_k(t))$ in the vineyard at level $t$ and the vine consisting of these points if closed (by identifying the vineyard at times $0$ and $2\pi$) will yield a circle for each $c$. 
To put it differently this will lead to a surgery as depicted in Figure \ref{fig:braid_to_ouroboros}. 
Finally, we also note that if $l=1$ there is an additional $1$-cycle, that is born much later than all the other $1$-cycles.

\begin{figure}[h!]
    \centering
    \includegraphics[width=0.75\linewidth]{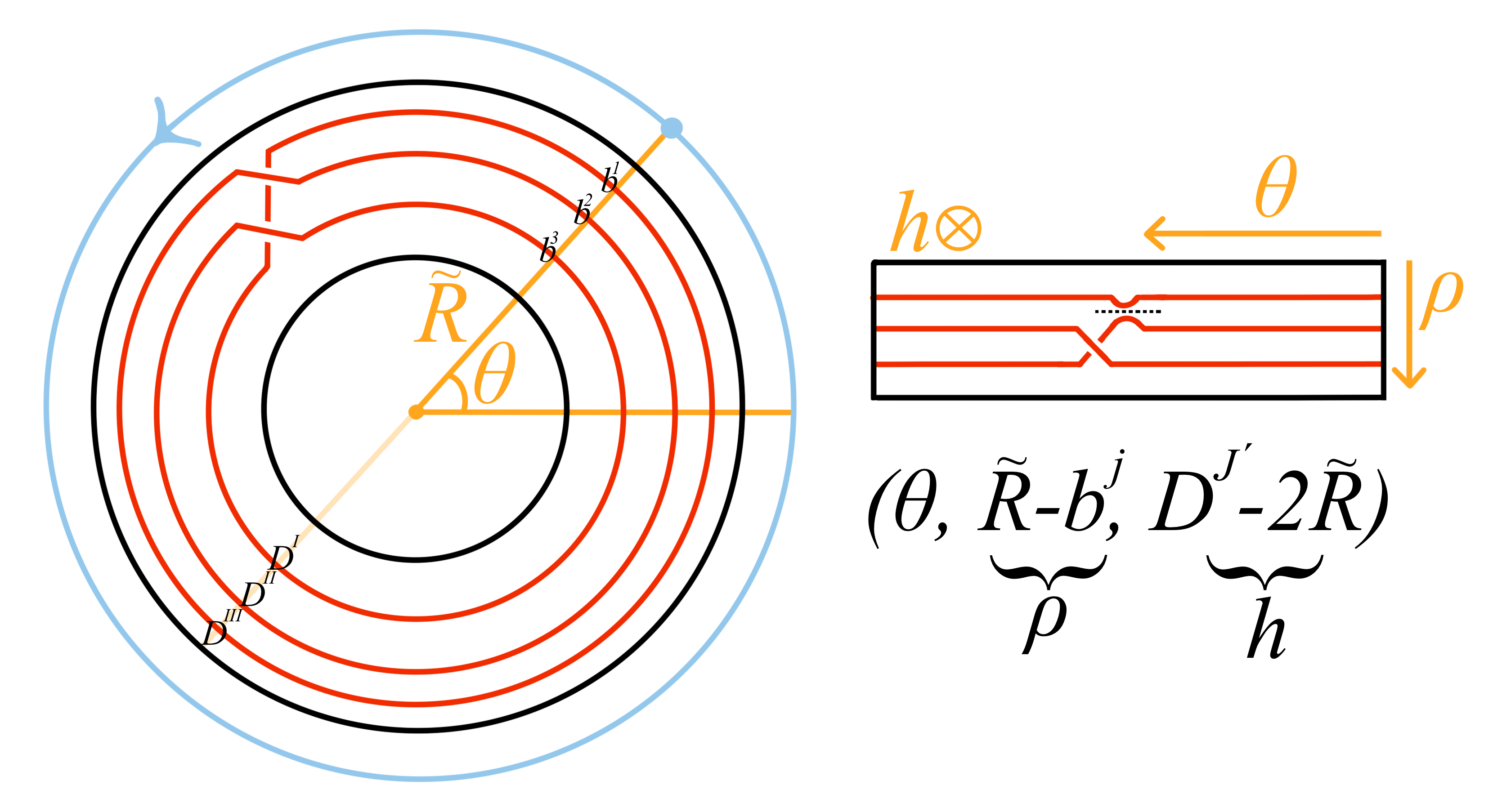}
    \caption{Simplified sketch: On the left, we see the particular closed braid that we called the Ouroboros. On the right, we see the $\theta$ and $\rho$ coordinates in the vineyard, where the $h$ coordinate is in the direction orthogonal to the plane and $\tilde{R}$ is the distance from the observation loop to the origin in the annulus. Here we identify $\rho$ with $\tilde{R}-b^j$ and $h$ with $D^{J'}-2\tilde{R}$, where $D^{J'}$ is the death time of the cycle born at $b^j$, corresponding to a strand as indicated in the figure.  
    The elder rule induces surgery indicated with a dashed line (as depicted on the right). Note that this sketch is simplified in that the annulus here is depicted in the plane, while the actual embedded is twisted as described previously.  %
    }
   \label{fig:braid_to_ouroboros}
\end{figure}

We further note that because vines are continuous (and even Lipschitz)  thanks to \cite{cohen2005stability}, the only thing which we need to worry about is the crossings, because if all the birth values are distinct the birth values give precisely the $\rho$ coordinates in Figure \ref{fig:braid_to_ouroboros}, death times corresponding to the coordinate $h$ do not matter.  

We recall that there is only one crossing per (angular) interval $[ \pi \frac{\tilde{j}}{8 (C+1)} , \pi \frac{\tilde{j}+1}{8 (C+1)}  ]$ where $\tilde{j} \in \{  8, 16, \dots , 8 C\}$ and no crossings in $[\pi+\frac{\pi}{8 (C+1)}, 2 \pi -\frac{\pi}{8 (C+1)}]$. Because the death times only matter for the crossing we can change the death times between crossings without creating topological problems. For a crossing in the interval $[ \pi \frac{\tilde{j}}{8 (C+1)} , \pi \frac{\tilde{j}+1}{8 (C+1)}  ]$ we dictate the death times by the geometry of the strands in the angular interval $[ \pi \frac{\tilde{j}-1}{8 (C+1)} +\pi , \pi \frac{\tilde{j}+2}{8 (C+1)} +\pi ]$. We do so by changing the $\tilde{\psi}$ coordinates of the strands in the interval, so that the death times for the interval $[ \pi \frac{\tilde{j}}{8 (C+1)} , \pi \frac{\tilde{j}+1}{8 (C+1)}  ]$ change, but the birth times in the interval  $[\pi+\frac{\pi}{8 (C+1)}, 2 \pi -\frac{\pi}{8 (C+1)}]$ remain the same, see Figure \ref{fig:psiManipulation}.  This in particular ensures that we do not introduce (extra) crossings in the vineyard that are not present in the closed braid we start with.

\begin{figure}[h!]
    \centering
    \includegraphics[width=0.75\linewidth]{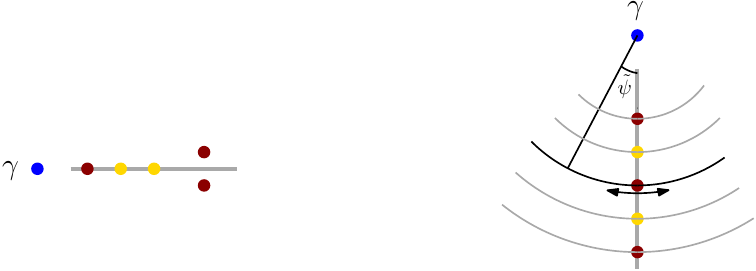}
    \caption{Manipulating the $\tilde{\psi}$ changes the Death times (vertical section of Figure \ref{fig:mainResult1}, not to scale). %
    }
    \label{fig:psiManipulation}
\end{figure}

Between these intervals we interpolate between the different geometries of strands (again by changing the $\tilde{\psi}$ coordinates), which we can do because as mentioned the death times do not influence the topology of the closure of the braid appearing in the persistence diagram.

Now let us consider a crossing of two strands in the interval $[ \pi \frac{\tilde{j}}{8 (C+1)} , \pi \frac{\tilde{j}+1}{8 (C+1)}  ]$ that are both not born first. 
As in intermezzo we denote with a little bit of abuse of notation by $b^{j,c}(t)$ and $b^{k,\tilde{c}}(t)$ (for $B$, or $b^{j,c}_0(t)$, $b^{k,\tilde{c}}_0(t)$,$b^{j,\tilde{c}}_l(t)$ and $b^{k,c}_l(t)$ respectively in the case of $\M$) both the birth times for these strands as well as the Morse critical points, where $c,\tilde{c}$ denote different connected components. We further stress that as in the intermezzo the $j$ and $k$ indices in $b^{j,c}(t)$ and $b^{k,\tilde{c}}(t)$ refer to order along the closed braid not the order of insertion. 
In the case where $l=1$ we further assume that neither $j$ nor $k$ corresponds to the birth with the very high birth value, that is the point where in non-extended persistence the second $1$-cycle that lives forever is born. 

We distinguish two different cases, one where $c = \tilde{c}$ and one where $c \neq \tilde{c}$. We start with the latter: We'll focus on the (somewhat simpler) $B$ case, as the $\M$ case is virtually identical. Because the birth times correspond to Morse critical points on different connected components changing order does not influence the pairing between Morse critical points and birth and death in the persistence diagrams. This means that as long as the death times are distinct in  $[ \pi \frac{\tilde{j}}{8 (C+1)} , \pi \frac{\tilde{j}+1}{8 (C+1)}  ]$ and the cycle born at $b^{j,c}(t)$ dies before the one born at $b^{k,\tilde{c}}(t)$ for one $t$ in this interval, then the cycle born at $b^{j,c}(t)$ dies before the one born at $b^{k,\tilde{c}}(t)$ for all $t$ in this interval and at continuous and distinct death times. Let us write $D^{J',c} (t)$ and $D^{K',\tilde{c}}(t)$ for the death times.  
This means that the persistence diagram (at level $t$ in the vineyard) contains the points $(b^{j,c}(t) ,D^{J',c} (t))$ and $(b^{k,\tilde{c}}(t) ,D^{K',\tilde{c}}(t))$, each of which locally describe a vine. This in turn implies that if (locally) $D^{J',c} (t)'> D^{K',\tilde{c}}(t)$  then the vine $( b^{j,c}(t) ,D^{J',c} (t))$ in the vineyard crosses under the vine $(b^{k,\tilde{c}}(t) ,D^{K',\tilde{c}}(t))$ and the reversed order of death corresponds to an under crossing. This is achieved (as mentioned) by manipulating the $\tilde{\psi}$ coordinate. An easy way to achieve this if the component $c$ passes under $\tilde{c}$ is to push all strands of $c$ in and all strands of $\tilde{c}$ out and the reverse for an over pass.  If one would like to repeat this discussion for $\M$ one only need to add a lower order $0$ or $l$ respectively. 

 We note that in the case where $l=1$, the birth with the very high birth value, that is the point where in non-extended persistence the second $1$-cycle that lives forever is born and was denoted by $m$ in the intermezzo, is separate (because of the high birth value) and therefore gives a disconnected circle in the closure of the vineyard.  

Now we consider the case there $c=\tilde{c}$. To simplify notation we will drop the index $c$ from the notation altogether. We again focus on the $B$ case, the case $\M$ is almost identical. 
If the death times satisfy \eqref{eq:DeathConditions} for each $t$ in  $[ \pi \frac{\tilde{j}}{8 (C+1)} , \pi \frac{\tilde{j}+1}{8 (C+1)}  ]$, then there are points $(b^j (t),D^{J-1}(t))$ and $(b^k(t) , D^K(t) )$ in the persistence diagram. We can ensure that \eqref{eq:DeathConditions} holds by changing the $\tilde{\psi}$ coordinates as before, see Figure \ref{fig:psiManipulation}.
Because the assumption \eqref{eq:DeathConditions} does not constrain the relative order of $b^j (t)$ and $b^k (t)$ nor the relative order of $D^{J-1}(t)$ and $D^{K}(t)$. This means that we can change the order of birth time (which occurs thanks to the crossing in the closed braid $B$) during the course of the interval $[ \pi \frac{\tilde{j}}{8 (C+1)} , \pi \frac{\tilde{j}+1}{8 (C+1)} ]$ and we can fix the relative order of $D^{J-1}(t)$ and $D^{K}(t)$ as needed (which we can again do by manipulating $\tilde{\psi}$, see Figure \ref{fig:psiManipulation}). If $D^{J-1}(t)>D^{K}(t)$ in  $[ \pi \frac{\tilde{j}}{8 (C+1)} , \pi \frac{\tilde{j}+1}{8 (C+1)}  ]$ then the $b_j$ vine crosses over the $b_k$ vine, while for the reversed order the $b_j$ vine crosses crosses under the $b_k$ vine. We refer to the appendix for an extensive example of this procedure in the case of the braid depicted in Figure \ref{fig:mainResult1}. 
This means that regardless of if we have an over or under crossing we locally push the strands with the critical points $D^{I}, D^{II}, \dots, D^{J-2}, D^{K+1}, D^{K+2} , \dots D^N$ in (towards the centre of the annulus) by a lot, and the strands with critical point $D^{J}, \dots  , D^{K-1}$ out (from the centre of the annulus) by a lot. This leaves the strands with $D^{J-1}$ and $D^{K}$ in the middle, and one may push the strand with $D^{J-1}$ a little out and the $D^{K}$ a little in if we want the $b^j$ vine segment to pass over the $b^k$ part of the vine, with the reverse pushing for the under crossing.    
This in particular shows that the braid $B$ can be faithfully reconstructed in the persistence diagram, although we also introduce some one extra loop per connected component of the link if $l\neq 1$ and two if $l=1$.

\end{itemize} 
The only thing left to remark is that the conditions of Lemma \ref{Lem:Step2} and Corollary \ref{Lem:Step4} can always be satisfied thanks to Remark \ref{rem:ConditionsAreEasyToSatisfy}. %
\end{proof}

\subsection{Vineyard monodromy} 

In this section we give the second main result of this paper, namely that every order of monodromy can be found in a vineyard and in any dimension $l$ of $H_l$, with $l\leq d-2$. This result follows almost immediately from Theorem \ref{thm:Link}, and answers one of the main open questions of Arya et al. \cite{Arya2024}, albeit in a slightly different context.

\thmMonodromy*

\begin{proof} The result follows by applying Theorem \ref{thm:Link} to the ($l$-offset) of the ouroborus knot, that is, the closure of the braid with $k+1$ strands {and $k$ over crossings,} as depicted in Figure \ref{fig:braid_to_ouroboros}. 
\end{proof}

\section{Conclusions and open questions}

We conclude by reflecting that the construction in this paper implies that vineyards are as topologically complex as one may hope or fear, opening up a number of interesting directions for future work.  We note that as stated in~\cite{onus2024shovingtubesshapesgives}, the radial transform presents some significant advantages over the more well-studied directional transform, and we feel it warrants greater study.  In particular, it is quite critical in our construction, and allowed for better understanding of monodromy and braiding in higher dimensions.

Theorem \ref{thm:Link} implies that comparing closed vineyards with a measure that completely reflects the topology is likely to be a difficult problem, computationally speaking. By this we mean the following: Suppose that we want to compare periodic phenomena using persistence, and thus are faced with closed vineyards. Now suppose that we want a distance between two closed vineyards that respects the topology. As these vineyards can contain arbitrary knots and links embedded in them, solving this problem would necessarily mean recognizing those structures within.
While the exact complexity of knot recognition is formally still open, in practice it has proven to be quite difficult, as it sits at the intersection of NP~\cite{Ichihara2023,hass1997algorithms} and co-NP~\cite{Lackenby2021} and is connected to several known hardness results~\cite{Koenig2021}.  In particular, it is known to be NP-hard to test unknottedness if the number of simplifying Reidemeister moves allowed while unknotting is specified as a part of the input~\cite{deMesmay2021}.

That said, the richness of structure present in vineyards may still allow for subtle distinctions and invariants to be calculated.  In fact, our results in some sense show that vineyards are more than the sum of parts, as the choice of base loop in our example is in fact critical for the resulting vineyard to have any braiding, and the resulting vineyard can exhibit quite complex behavior.
However, these subtleties come at a high computational cost, and finding a compromise between topological fidelity and computational complexity will be an important challenge for future work.

We conclude with a few specific open questions:
\begin{itemize}
\item[-] Are there good combinations of link invariants and distances, like the Wasserstein distance, that still capture a lot of the topology, provide a practical similarity measure, and are relatively easy to compute?  If so, could these invariants prove useful in topological data analysis, given the rising use of the persistent homology transform in real world applications?
\item[-] Are there geometric conditions (that are easy to verify) on a pair of closed vineyards that imply that the closed vineyards are isomorphic?
\item[-] Could tools from topological data analysis, and in particular statistical work on vineyards~\cite{Munch2015}, provide new practical insights into knot and link recognition?
\item[-] On the application side, it would be interesting to find even simple non-trivial knots or links (like the trefoil knot, Hopf link, or  Borromean rings) in persistence diagrams of real data, as there is some indication of knots present in proteins~\cite{Letscher2012}. If we may add a little speculation, the authors would start looking in persistence diagrams of biological systems containing circular DNA \cite{neill2024enzymatic}.  
\end{itemize}

\section*{Acknowledgments}
We thank the reviewers of both SODA and ATMCS for their comments, which improved the exposition of our paper. 
We thank Kate Turner for discussion and Cl{\'e}ment Maria for pointing out that Alexander's theorem was already (well) known. 
Mathijs Wintraecken would like to express his gratitude to the administrative support he received from University of Notre Dame  during his visit and from Sophie Honnorat and Stephanie Verdonck at Inria in general. 

This work has been supported by the ANR grant StratMesh, ANR-24-CE48-1899, by NSF award 2444309, and the welcome
package from IDEX of the Université Côte d’Azur, ANR-15-IDEX-01.

\phantomsection
\addcontentsline{toc}{section}{Bibliography}
\bibliography{bib}

\begin{thebibliography}{10}

\bibitem{Agarwal2006}
Pankaj~K. Agarwal, Herbert Edelsbrunner, John Harer, and Yusu Wang.
\newblock Extreme elevation on a 2-manifold.
\newblock {\em Discrete \& Computational Geometry}, 36(4):553–572, September 2006.
\newblock URL: \url{http://dx.doi.org/10.1007/s00454-006-1265-8}, \href {https://doi.org/10.1007/s00454-006-1265-8} {\path{doi:10.1007/s00454-006-1265-8}}.

\bibitem{Alexander1923}
J.~W. Alexander.
\newblock A lemma on systems of knotted curves.
\newblock {\em Proceedings of the National Academy of Sciences}, 9(3):93–95, March 1923.
\newblock URL: \url{http://dx.doi.org/10.1073/pnas.9.3.93}, \href {https://doi.org/10.1073/pnas.9.3.93} {\path{doi:10.1073/pnas.9.3.93}}.

\bibitem{Alexander1926}
J.~W. Alexander and G.~B. Briggs.
\newblock On types of knotted curves.
\newblock {\em The Annals of Mathematics}, 28(1/4):562, 1926.
\newblock URL: \url{http://dx.doi.org/10.2307/1968399}, \href {https://doi.org/10.2307/1968399} {\path{doi:10.2307/1968399}}.

\bibitem{Arya2024}
Shreya Arya, Barbara Giunti, Abigail Hickok, Lida Kanari, Sarah McGuire, and Katharine Turner.
\newblock Decomposing the persistent homology transform of star-shaped objects, 2024.
\newblock URL: \url{https://arxiv.org/abs/2408.14995}, \href {https://doi.org/10.48550/ARXIV.2408.14995} {\path{doi:10.48550/ARXIV.2408.14995}}.

\bibitem{Birman2005}
Joan~S. Birman and Tara~E. Brendle.
\newblock {\em Braids}, page 19–103.
\newblock Elsevier, 2005.
\newblock URL: \url{http://dx.doi.org/10.1016/B978-044451452-3/50003-4}, \href {https://doi.org/10.1016/b978-044451452-3/50003-4} {\path{doi:10.1016/b978-044451452-3/50003-4}}.

\bibitem{birman1999new}
Joan~S Birman and Michael~D Hirsch.
\newblock A new algorithm for recognizing the unknot.
\newblock {\em Geometry \& Topology}, 2(1):175--220, 1999.

\bibitem{Burton2020}
Benjamin~A. Burton.
\newblock {The Next 350 Million Knots}.
\newblock In Sergio Cabello and Danny~Z. Chen, editors, {\em 36th International Symposium on Computational Geometry (SoCG 2020)}, volume 164 of {\em Leibniz International Proceedings in Informatics (LIPIcs)}, pages 25:1--25:17, Dagstuhl, Germany, 2020. Schloss Dagstuhl -- Leibniz-Zentrum f{\"u}r Informatik.
\newblock URL: \url{https://drops.dagstuhl.de/entities/document/10.4230/LIPIcs.SoCG.2020.25}, \href {https://doi.org/10.4230/LIPIcs.SoCG.2020.25} {\path{doi:10.4230/LIPIcs.SoCG.2020.25}}.

\bibitem{Carriere2020}
Mathieu Carri\`{e}re and Andrew~J. Blumberg.
\newblock Multiparameter persistence images for topological machine learning.
\newblock In {\em Proceedings of the 34th International Conference on Neural Information Processing Systems}, NIPS '20, Red Hook, NY, USA, 2020. Curran Associates Inc.

\bibitem{Cerri2013}
Andrea Cerri, Marc Ethier, and Patrizio Frosini.
\newblock A study of monodromy in the computation of multidimensional persistence.
\newblock In {\em Discrete Geometry for Computer Imagery}, pages 192--202, Berlin, Heidelberg, 2013. Springer Berlin Heidelberg.

\bibitem{Cipriani2023}
Alessandra Cipriani, Christian Hirsch, and Martina Vittorietti.
\newblock Topology-based goodness-of-fit tests for sliced spatial data.
\newblock {\em Computational Statistics \& Data Analysis}, 179:107655, March 2023.
\newblock URL: \url{http://dx.doi.org/10.1016/j.csda.2022.107655}, \href {https://doi.org/10.1016/j.csda.2022.107655} {\path{doi:10.1016/j.csda.2022.107655}}.

\bibitem{CogolludoAgustn2011}
José~Ignacio Cogolludo-Agustín.
\newblock Braid monodromy of algebraic curves.
\newblock {\em Annales mathématiques Blaise Pascal}, 18(1):141–209, 2011.
\newblock URL: \url{http://dx.doi.org/10.5802/ambp.295}, \href {https://doi.org/10.5802/ambp.295} {\path{doi:10.5802/ambp.295}}.

\bibitem{Cohen1997}
Daniel~C. Cohen and Alexander~I. Suciu.
\newblock The braid monodromy of plane algebraic curves and hyperplane arrangements.
\newblock {\em Commentarii Mathematici Helvetici}, 72(2):285–315, June 1997.
\newblock URL: \url{http://dx.doi.org/10.1007/s000140050017}, \href {https://doi.org/10.1007/s000140050017} {\path{doi:10.1007/s000140050017}}.

\bibitem{cohen2005stability}
David Cohen-Steiner, Herbert Edelsbrunner, and John Harer.
\newblock Stability of persistence diagrams.
\newblock In {\em Proceedings of the twenty-first annual symposium on Computational geometry}, pages 263--271, 2005.

\bibitem{CohenSteiner2008}
David Cohen-Steiner, Herbert Edelsbrunner, and John Harer.
\newblock Extending persistence using poincaré and lefschetz duality.
\newblock {\em Foundations of Computational Mathematics}, 9(1):79–103, April 2008.
\newblock URL: \url{http://dx.doi.org/10.1007/s10208-008-9027-z}, \href {https://doi.org/10.1007/s10208-008-9027-z} {\path{doi:10.1007/s10208-008-9027-z}}.

\bibitem{CohenSteiner2006}
David Cohen-Steiner, Herbert Edelsbrunner, and Dmitriy Morozov.
\newblock Vines and vineyards by updating persistence in linear time.
\newblock In {\em Proceedings of the twenty-second annual symposium on Computational geometry}, SoCG06, page 119–126. ACM, June 2006.
\newblock URL: \url{http://dx.doi.org/10.1145/1137856.1137877}, \href {https://doi.org/10.1145/1137856.1137877} {\path{doi:10.1145/1137856.1137877}}.

\bibitem{CrawleyBoevey2015}
William Crawley-Boevey.
\newblock Decomposition of pointwise finite-dimensional persistence modules.
\newblock {\em Journal of Algebra and Its Applications}, 14(05):1550066, March 2015.
\newblock URL: \url{http://dx.doi.org/10.1142/S0219498815500668}, \href {https://doi.org/10.1142/s0219498815500668} {\path{doi:10.1142/s0219498815500668}}.

\bibitem{deMesmay2021}
Arnaud de~Mesmay, Yo’av Rieck, Eric Sedgwick, and Martin Tancer.
\newblock The unbearable hardness of unknotting.
\newblock {\em Advances in Mathematics}, 381:107648, April 2021.
\newblock URL: \url{http://dx.doi.org/10.1016/j.aim.2021.107648}, \href {https://doi.org/10.1016/j.aim.2021.107648} {\path{doi:10.1016/j.aim.2021.107648}}.

\bibitem{Dehn1911}
M.~Dehn.
\newblock {\"U}ber unendliche diskontinuierliche {G}ruppen.
\newblock {\em Mathematische Annalen}, 71(1):116–144, mar 1911.
\newblock URL: \url{http://dx.doi.org/10.1007/BF01456932}, \href {https://doi.org/10.1007/bf01456932} {\path{doi:10.1007/bf01456932}}.

\bibitem{Dehn1912}
M.~Dehn.
\newblock Transformation der {K}urven auf zweiseitigen {F}l{\"a}chen.
\newblock {\em Mathematische Annalen}, 72(3):413–421, sep 1912.
\newblock URL: \url{http://dx.doi.org/10.1007/BF01456725}, \href {https://doi.org/10.1007/bf01456725} {\path{doi:10.1007/bf01456725}}.

\bibitem{Dey2022}
Tamal~Krishna Dey and Yusu Wang.
\newblock {\em Computational Topology for Data Analysis}.
\newblock Cambridge University Press, February 2022.
\newblock URL: \url{http://dx.doi.org/10.1017/9781009099950}, \href {https://doi.org/10.1017/9781009099950} {\path{doi:10.1017/9781009099950}}.

\bibitem{ebeling2005monodromy}
Wolfgang Ebeling.
\newblock Monodromy, 2005.
\newblock URL: \url{https://arxiv.org/abs/math/0507171}, \href {https://arxiv.org/abs/math/0507171} {\path{arXiv:math/0507171}}.

\bibitem{Edelsbrunner2002}
Edelsbrunner, Letscher, and Zomorodian.
\newblock Topological persistence and simplification.
\newblock {\em Discrete \& Computational Geometry}, 28(4):511–533, November 2002.
\newblock URL: \url{http://dx.doi.org/10.1007/s00454-002-2885-2}, \href {https://doi.org/10.1007/s00454-002-2885-2} {\path{doi:10.1007/s00454-002-2885-2}}.

\bibitem{Federer}
H.~Federer.
\newblock Curvature measures.
\newblock {\em Transactions of the American Mathematical Society}, 93:418--491, 1959.
\newblock \href {https://doi.org/10.1090/S0002-9947-1959-0110078-1} {\path{doi:10.1090/S0002-9947-1959-0110078-1}}.

\bibitem{Frosini1990}
Patrizio Frosini.
\newblock A distance for similarity classes of submanifolds of a euclidean space.
\newblock {\em Bulletin of the Australian Mathematical Society}, 42(3):407–415, December 1990.
\newblock URL: \url{http://dx.doi.org/10.1017/S0004972700028574}, \href {https://doi.org/10.1017/s0004972700028574} {\path{doi:10.1017/s0004972700028574}}.

\bibitem{hass1997algorithms}
Joel Hass.
\newblock Algorithms for recognizing knots and 3-manifolds.
\newblock {\em arXiv preprint math/9712269}, 1997.

\bibitem{hatcher2002algebraic}
Allen Hatcher.
\newblock {\em Algebraic Topology}.
\newblock Cambridge University Press, 2002.

\bibitem{Hickok2022b}
Abigail Hickok.
\newblock Persistence diagram bundles: A multidimensional generalization of vineyards, 2022.
\newblock URL: \url{https://arxiv.org/abs/2210.05124}, \href {https://doi.org/10.48550/ARXIV.2210.05124} {\path{doi:10.48550/ARXIV.2210.05124}}.

\bibitem{Ichihara2023}
Kazuhiro Ichihara, Yuya Nishimura, and Seiichi Tani.
\newblock The computational complexity of classical knot recognition.
\newblock {\em Journal of Knot Theory and Its Ramifications}, 32(11), October 2023.
\newblock URL: \url{http://dx.doi.org/10.1142/S0218216523500694}, \href {https://doi.org/10.1142/s0218216523500694} {\path{doi:10.1142/s0218216523500694}}.

\bibitem{Koenig2021}
Dale Koenig and Anastasiia Tsvietkova.
\newblock Np–hard problems naturally arising in knot theory.
\newblock {\em Transactions of the American Mathematical Society, Series B}, 8(15):420–441, May 2021.
\newblock URL: \url{http://dx.doi.org/10.1090/btran/71}, \href {https://doi.org/10.1090/btran/71} {\path{doi:10.1090/btran/71}}.

\bibitem{Lackenby2021}
Marc Lackenby.
\newblock The efficient certification of knottedness and thurston norm.
\newblock {\em Advances in Mathematics}, 387:107796, August 2021.
\newblock URL: \url{http://dx.doi.org/10.1016/j.aim.2021.107796}, \href {https://doi.org/10.1016/j.aim.2021.107796} {\path{doi:10.1016/j.aim.2021.107796}}.

\bibitem{Letscher2012}
David Letscher.
\newblock On persistent homotopy, knotted complexes and the alexander module.
\newblock In {\em Proceedings of the 3rd Innovations in Theoretical Computer Science Conference}, ITCS ’12, page 428–441. ACM, January 2012.
\newblock URL: \url{http://dx.doi.org/10.1145/2090236.2090270}, \href {https://doi.org/10.1145/2090236.2090270} {\path{doi:10.1145/2090236.2090270}}.

\bibitem{milnor1963morse}
John~Willard Milnor.
\newblock {\em Morse theory}.
\newblock Number~51 in Annals of Mathematics Studies. Princeton university press, 1963.

\bibitem{dionysus}
Dmitriy Morozov.
\newblock Dionysis 2: A library for computing persistent homology.
\newblock URL: \url{https://mrzv.org/software/dionysus2/}.

\bibitem{Munch2015}
Elizabeth Munch, Katharine Turner, Paul Bendich, Sayan Mukherjee, Jonathan Mattingly, and John Harer.
\newblock Probabilistic fréchet means for time varying persistence diagrams.
\newblock {\em Electronic Journal of Statistics}, 9(1), January 2015.
\newblock URL: \url{http://dx.doi.org/10.1214/15-EJS1030}, \href {https://doi.org/10.1214/15-ejs1030} {\path{doi:10.1214/15-ejs1030}}.

\bibitem{munkres2018elements}
James~R Munkres.
\newblock {\em Elements of algebraic topology}.
\newblock CRC press, 2018.

\bibitem{neill2024enzymatic}
Philip Neill, Natalie Crist, Ryan McGorty, and Rae Robertson-Anderson.
\newblock Enzymatic cleaving of entangled {DNA} rings drives scale-dependent rheological trajectories.
\newblock {\em Soft Matter}, 20(12):2750--2766, 2024.

\bibitem{AATRNsara}
Applied Algebraic~Topology Network.
\newblock Sara scaramuccia (04/23/25): Monodromy in bi-parameter persistence modules, 2025.
\newblock URL: \url{https://youtu.be/14dSHjeFUoI?si=XUdRBLtG1ImwrR88}.

\bibitem{onus2024shovingtubesshapesgives}
Adam Onus, Nina Otter, and Renata Turkes.
\newblock Shoving tubes through shapes gives a sufficient and efficient shape statistic, 2024.
\newblock URL: \url{https://arxiv.org/abs/2412.18452}, \href {https://arxiv.org/abs/2412.18452} {\path{arXiv:2412.18452}}.

\bibitem{Oudot2015}
Steve Oudot.
\newblock {\em Persistence Theory: From Quiver Representations to Data Analysis}.
\newblock American Mathematical Society, December 2015.
\newblock URL: \url{http://dx.doi.org/10.1090/surv/209}, \href {https://doi.org/10.1090/surv/209} {\path{doi:10.1090/surv/209}}.

\bibitem{Reidemeister1927}
Kurt Reidemeister.
\newblock Elementare begr\"{u}ndung der knotentheorie.
\newblock {\em Abhandlungen aus dem Mathematischen Seminar der Universit\"{a}t Hamburg}, 5(1):24–32, December 1927.
\newblock URL: \url{http://dx.doi.org/10.1007/BF02952507}, \href {https://doi.org/10.1007/bf02952507} {\path{doi:10.1007/bf02952507}}.

\bibitem{Robins1999}
Vanessa Robins.
\newblock Towards computing homology from approximations.
\newblock {\em Topology Proceedings}, 24, 01 1999.

\bibitem{Rubinstein92}
J.~H. Rubinstein.
\newblock The solution to the recognition problem for $s^3$.
\newblock Lectures in Haifa, Israel, 1992.

\bibitem{salter2023stratifiedbraidgroupsmonodromy}
Nick Salter.
\newblock Stratified braid groups: monodromy, 2023.
\newblock URL: \url{https://arxiv.org/abs/2304.04627}, \href {https://arxiv.org/abs/2304.04627} {\path{arXiv:2304.04627}}.

\bibitem{Salter2024}
Nick Salter.
\newblock Monodromy of stratified braid groups, ii.
\newblock {\em Research in the Mathematical Sciences}, 11(4), October 2024.
\newblock URL: \url{http://dx.doi.org/10.1007/s40687-024-00477-4}, \href {https://doi.org/10.1007/s40687-024-00477-4} {\path{doi:10.1007/s40687-024-00477-4}}.

\bibitem{ElizabethsThesis}
Elizabeth~R. Stephenson.
\newblock Generalizing medial axes with homology switches.
\newblock Master's thesis, Institute of Science and Technology Austria, 2023.
\newblock Available at \url{https://doi.org/10.15479/at:ista:14226}.

\bibitem{thompson1994thin}
Abigail Thompson.
\newblock Thin position and the recognition problem for $\mathbf{S}^{3}$.
\newblock {\em Mathematical Research Letters}, 1(5):613--630, 1994.

\bibitem{turner2023representing}
Katharine Turner.
\newblock Representing vineyard modules.
\newblock {\em arXiv preprint arXiv:2307.06020}, 2023.

\bibitem{Vogel1990}
Pierre Vogel.
\newblock Representation of links by braids: A new algorithm.
\newblock {\em Commentarii mathematici Helvetici}, 65(1):104--113, 1990.
\newblock URL: \url{http://eudml.org/doc/140186}.

\bibitem{williams2024new}
Virginia~Vassilevska Williams, Yinzhan Xu, Zixuan Xu, and Renfei Zhou.
\newblock New bounds for matrix multiplication: from alpha to omega.
\newblock In {\em Proceedings of the 2024 Annual ACM-SIAM Symposium on Discrete Algorithms (SODA)}, pages 3792--3835. SIAM, 2024.

\bibitem{Xian2022}
Lu~Xian, Henry Adams, Chad~M. Topaz, and Lori Ziegelmeier.
\newblock Capturing dynamics of time-varying data via topology.
\newblock {\em Foundations of Data Science}, 4(1):1, 2022.
\newblock URL: \url{http://dx.doi.org/10.3934/fods.2021033}, \href {https://doi.org/10.3934/fods.2021033} {\path{doi:10.3934/fods.2021033}}.

\bibitem{yamada1987index}
Shuji Yamada.
\newblock The minimal number of seifert circles equals the braid index of a link.
\newblock {\em Inventiones Mathematicae}, 89(2):347–356, Jun 1987.
\newblock \href {https://doi.org/10.1007/bf01389082} {\path{doi:10.1007/bf01389082}}.

\bibitem{Yoo2016}
Jaejun Yoo, Eun~Young Kim, Yong~Min Ahn, and Jong~Chul Ye.
\newblock Topological persistence vineyard for dynamic functional brain connectivity during resting and gaming stages.
\newblock {\em Journal of Neuroscience Methods}, 267:1–13, July 2016.
\newblock URL: \url{http://dx.doi.org/10.1016/j.jneumeth.2016.04.001}, \href {https://doi.org/10.1016/j.jneumeth.2016.04.001} {\path{doi:10.1016/j.jneumeth.2016.04.001}}.

\end{thebibliography}

 \appendix

 \section{Example}
 \label{app:Example}

As an example of pushing outwards and inwards of the strands by manipulating the $\tilde{\psi}$ coordinates in Step 4 of the proof of Theorem \ref{thm:Link} we discuss the example shown in Figure \ref{fig:mainResult1} in detail. We number the crossings $c_1, \dots , c_{13}$ as indicated in Figure \ref{fig:mainResultCrossings}, see also Figure \ref{fig:pretty_trefcircle}. The manipulation proceeds as below. We emphasize that all of the death values associated to points near a crossing need to be generic, that is, there are no two identical death values. We will not repeat this for every crossing.

\begin{figure}[h!]
    \centering
    \includegraphics[width=0.5\linewidth]{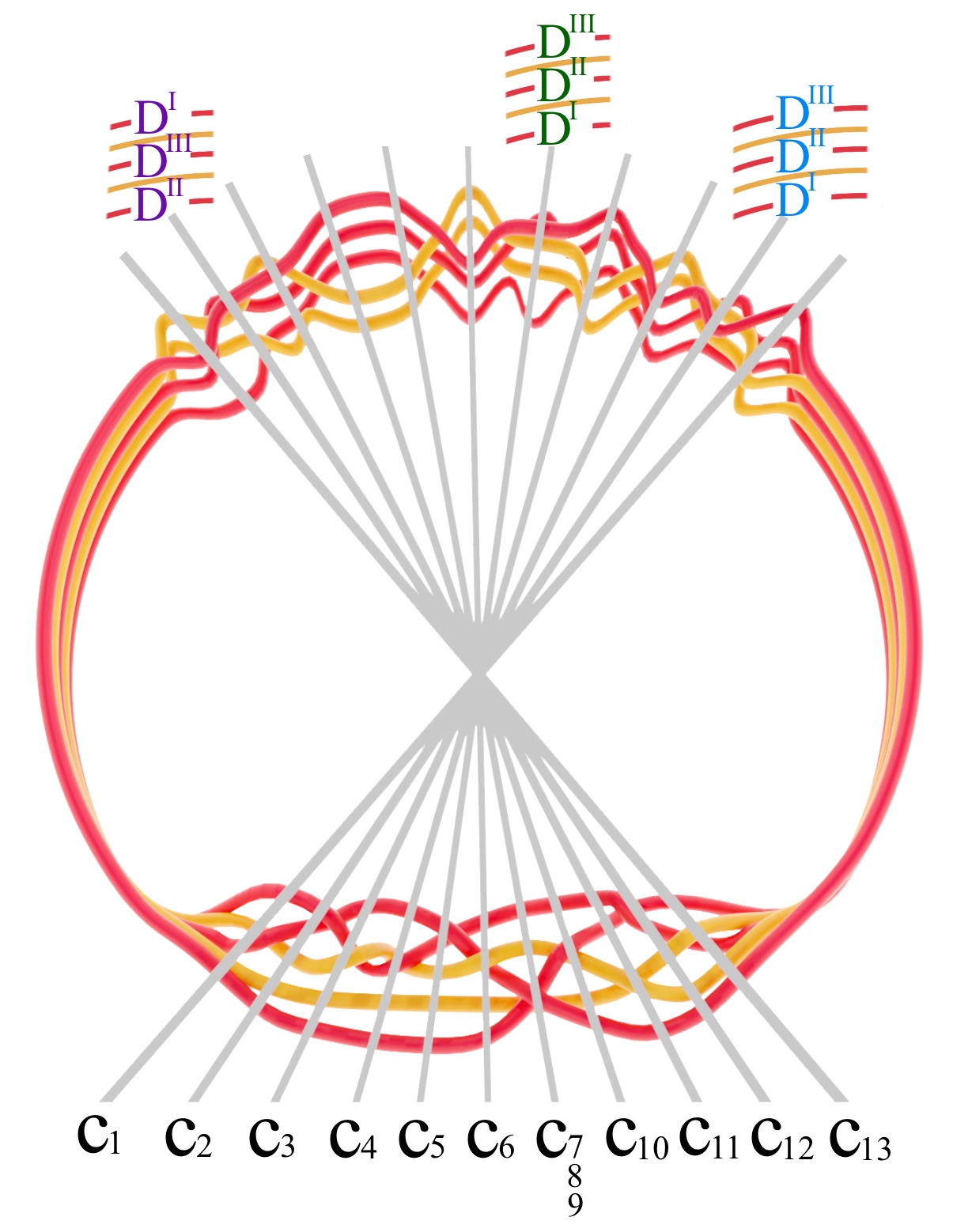}
    \caption{The crossings and the specific correspondence between the deaths or critical points and the strands indicated. \label{fig:mainResultCrossings} } 
\end{figure}

\begin{itemize} 
\item At $c_1$ yellow crosses under red, hence we push all red strands out and all yellow ones in. We will refer to this as yellow before red. 
\item At $c_2$ the strand on which $b^2$ lies crosses under the strand with $b^3$, in accordance with \eqref{eq:DeathConditions} we push $D^{III}$ in and $D^{II}$ out so that $D^{I}$ lies in between (blue in the figure), what you do with the yellow strands doesn't matter (as long as it is generic).  The way that we push is also indicated in blue in the figure.  
\item At $c_3$ red crosses under yellow, hence we push the yellow strands out and the red ones in, i.e. red before yellow.
\item At $c_4$ red crosses over yellow, hence push the red strands out and the yellow ones in, i.e. yellow before red. 
\item Similarly to the crossing at $c_2$ (but at $c_5$, the strand with $b^2$ crosses under the strand $b^3$), at $c_5$ we push $D^I$ in and push $D^{II}$ out so that $D^{III}$ (green in the figure) lies in the middle (there is no condition on the yellow strands except genericity). Indicated in green in the figure.  
\item At $c_6$ red goes before yellow.
\item  The crossings $c_7$, $c_8$, and $c_9$ almost coincide in the figure, however, the red crosses in all cases over yellow, so that yellow needs to go before red\footnote{Strictly speaking it is not necessary that yellow needs to go before red as long as you choose consistently for both crossings, because at the red crossing, the first birth changes from one strand to another, which leads to surgery so that the outer strand disconnects, which means that after a Reidemeister II move you are fine. However, to fit with the text in the proof, it is best to have yellow before red.}, because the first birth in red exchanges strand the coupling is automatic and leads to a disconnected component (i.e. surgery is performed).  
\item At $c_{10}$ yellow goes before red.
\item At $c_{11}$ the order of the strands really doesn't matter as long as all death values are distinct (generic), because this is another first birth interchange (on yellow this time). 
\item The crossing at $c_{12}$ is again similar to the crossings at $c_5$ with the strand with $b^2$ crossing under $b^3$. We stress that because we label the critical points along the braid in order along the braid from the first birth, and the first birth has exchanged strand, the strand that contains $b^2$ is not the same strand as the one that contained $b^2$ at $c_5$ (where by the same we mean identification via shortest paths on the link). The strands that are associated to $D^{I}$, $D^{II}$, and $D^{III}$ change as a consequence as well, see Figure \ref{fig:mainResultCrossings} (purple). This having been said, we follow the same procedure as at $c_{5}$:
At $c_{12}$ push $D^{III}$ out and $D^{II}$ in, so that $D^{I}$ ends up in the middle (there is no condition on the yellow strands except genericity). Indicated in purple in the figure. 
\item At $c_{13}$ red goes before yellow. 
\end{itemize}

\begin{figure}

\centering
\includegraphics[width=.85\textwidth]{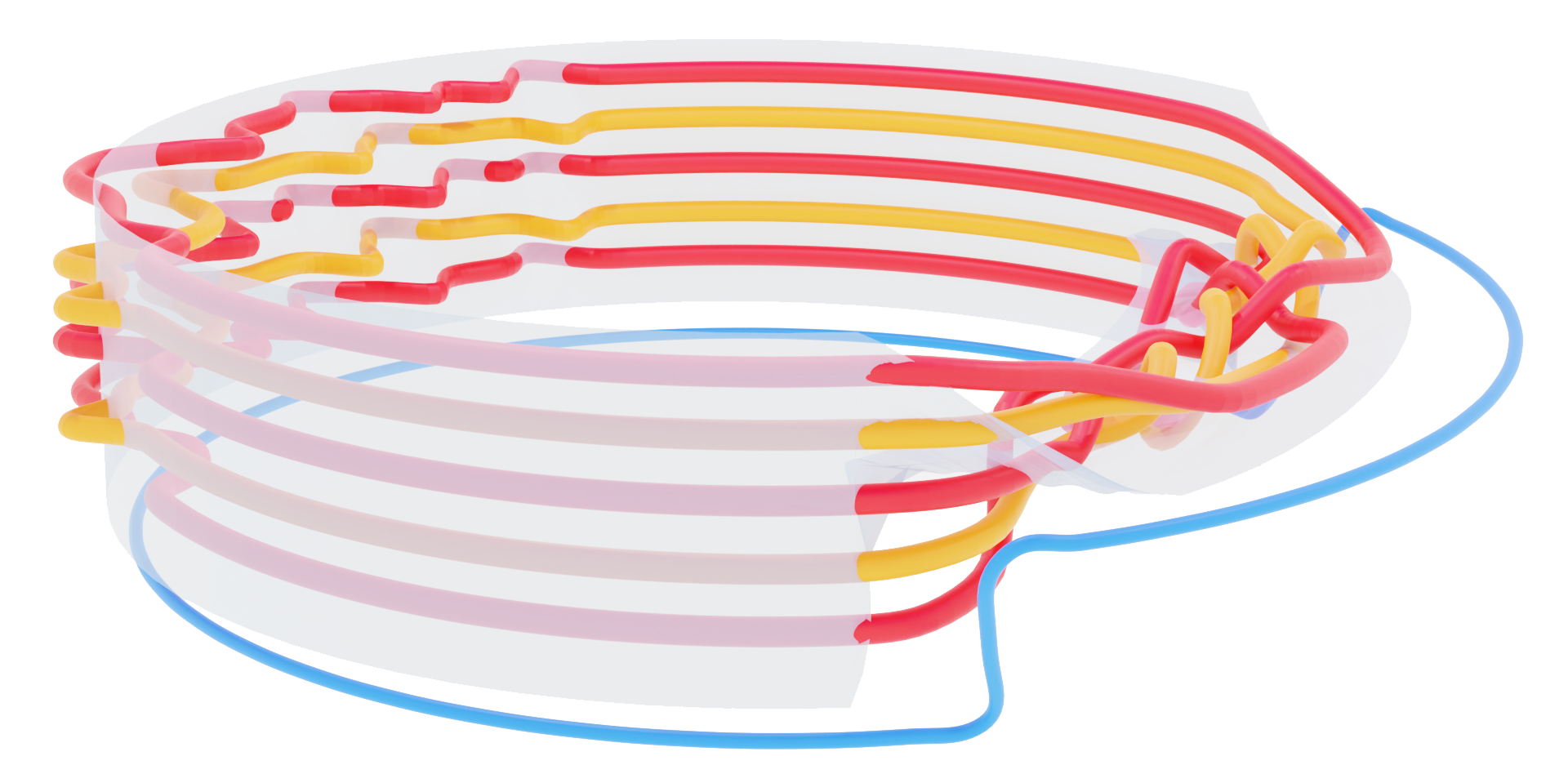}
\includegraphics[width=.85\textwidth]{figures/appendix_front_view_2.png}
\caption{Different views of the Example. Top: together with observation loop (blue) and twisted annulus (transparent gray). Bottom: front view of just the manipulated trefoil-circle link. \label{fig:pretty_trefcircle}}
\end{figure}

\section{Monodromy touching the diagonal} 
\label{appendix:diagonal}

If we drop the assumption that the vines do not touch the diagonal, but still assume that none of the persistence diagrams contain points with higher multiplicity, it is possible to formulate a notion of monodromy that is more equivalent to that contained in prior work~\cite{Arya2024}. 
To make sense of this definition given below and to provide an alternative formulation to the one in our main paper, we need the following observation: Let $\tilde{\mathcal{V}}_{\max} (t)$ be a vine defined on its  maximal domain $(t_{\min},t_{\max})$, where $t_{\min},t_{\max} \in \mathbb{R} \cup \{ \pm \infty\}$, that contains the point $\mathcal{V} (0)$. 
More formally $\tilde{\mathcal{V}}_{\max} (t)$ is the restriction to the maximal open interval such that the lift of the curve $\gamma: \mathbb{R} \to \mathbb{R} / 2 \pi \mathbb{Z}: t \mapsto t \; \textrm{mod} \; 2 \pi$, contains no limit points on the diagonal for this open interval. 

If we are given a vine $\mathcal{V}(t)$ defined on $[0,2 \pi]$ we call the vine  $\tilde{\mathcal{V}}_{\max} (t)$ defined on its maximal domain that coincides with $\mathcal{V}(t)$ on the interval $[0,2 \pi]$, the maximal extension  of $\mathcal{V}(t)$. If the domain of a maximal extension is finite, then we say that the vine has a finite maximal extension. 

If the number of points in each persistence diagram is finite (as we assume), then either both $t_{\min}$ and $t_{\max}$ are (plus/minus) infinite or neither of them are. This is clear because if one of them is infinite we must return to the same point in the persistence diagram after some time $2 \pi k$ (since the number of points in the persistence diagram is finite), in which case $\tilde{\mathcal{V}} (t)$  is periodic and is defined on  $(-\infty ,\infty)$. 

Using this definition we have the following reformulation of the order of the monodromy in the setting of Definition~\ref{Def:Monodromy1}:
\begin{remark}\label{rem:ReformulationMonodromy1}
Assume that there are no points of higher multiplicity in $\Per_l (d_\mathbb{E}(\cdot, \gamma(t) )|_\M)$, for all $\gamma(t)$ and by extension the vines are non-intersecting curves. Assume moreover that all of the vines are disjoint from the diagonal.
The vine $\mathcal{V} (t)$  exhibits monodromy of order $k$ if and only if $k$ is the smallest positive integer such that $\tilde{\mathcal{V}}_{\max} (0)=\tilde{\mathcal{V}}_{\max} (2 \pi k)$, where $\tilde{\mathcal{V}}_{\max} (t)$  is the maximal extension of $\mathcal{V} (t)$.   
\end{remark}

In this setting, that is, where we allow vines to contain points in the diagonal (as limit points), 
the goal again is to glue together in a way that we have the exhibited periodicity of monodromy as an integer.
More formally:

\begin{definition}
Assume that there are no points of multiplicity higher than $1$ in $\Per_l (d_\mathbb{E}(\cdot, \gamma(t) )|_\M)$, for all $\gamma(t)$ and by extension the vines are non-intersecting curves. We'll follow the same notation as above. 
\newline 
\textbf{Individual vine monodromy}
Let $\mathcal{V} (t)$ be the vine and $\tilde{\mathcal{V}}_{\max} (t)$ its maximal extension. If its maximal domain is $(-\infty ,\infty)$, then the vine does not touch the diagonal and Definition \ref{Def:Monodromy1} and the reformulation in Remark \ref{rem:ReformulationMonodromy1} both apply. 
Let us now assume that $t_{\min},t_{\max}$ are finite. By reparametrizing we can assume (without loss of generality) that $t_{\min}=0$. The order of monodromy is now defined as the smallest positive integer $k$, such that $t_{\max} <2 \pi k$. As before, we'll say that the monodromy is trivial if $k=1$. 

We'll also define the completion along the diagonal of a vine $\mathcal{V}$. Let $\tilde{\mathcal{V}}_{\max}: (t_{\min},t_{\max}) \to \Per_l (d_\mathbb{E}(\cdot, \gamma(t) )|_\M)$ be the maximal extension of $\mathcal{V}$, and write 
\begin{align} 
v_{\min} = \lim_{t \searrow t_{\min} } \tilde{\mathcal{V}}_{\max} (t) 
\nonumber
\\
v_{\max} = \lim_{t \nearrow t_{\max} } \tilde{\mathcal{V}}_{\max} (t) 
\nonumber 
\end{align} 
for the two limit points of the vine on the diagonal. 
Letting $k$ be the order of monodromy of $\mathcal{V}$, we define the completing diagonal vine $\mathcal{V}_D$ as 
\begin{align} 
\mathcal{V}_D :& (t_{\max} , t_{\min} + 2 \pi k) \to  \Per_l (d_\mathbb{E}(\cdot, \gamma(t) )|_\M)
\nonumber 
\\
& t \mapsto 
\left (t ,v_{\max} \left (1 -\frac{t-t_{\max}}{t_{\min} - t_{\max} + 2 \pi k} \right) + v_{\min} \frac{t-t_{\max}}{t_{\min} - t_{\max} + 2 \pi k}   \right) .
\nonumber 
\end{align} 
The start and end points of the concatenation $\mathcal{V}_D \circ \tilde{\mathcal{V}}_{\max}$ coincide and therefore by identifying $t_{\min}$ and $t_{\min} + 2 \pi k$ we can consider this to be a map on $\mathbb{R} / 2 \pi k \mathbb{Z}$. The set $\mathbb{R} / 2 \pi k \mathbb{Z}$  can be viewed as a $k$ cover of the circle $\mathbb{R} / 2 \pi \mathbb{Z}$, where we identify $\mathbb{R} / 2 \pi \mathbb{Z}$ with the loop $\gamma$. Composing with this cover map gives a map $\mathcal{V}_C : \gamma(t) \mapsto  \Per_l (d_\mathbb{E}(\cdot, \gamma(t) )|_\M)$, which exhibits monodromy in the way we defined above. We call $\mathcal{V}_C$ the completion of the vine $\mathcal{V}$.  
\newline
\textbf{Vineyard monodromy}
We now call the vines whose maximal extension do not intersect the diagonal (or equivalently those whose maximal domain is $(-\infty, \infty)$) non-rooted vines, while we call those that intersect the diagonal (or equivalently those whose maximal domain has finite length) rooted vines.  
We write $k_1, \dots , k_n$ for the orders of monodromy of non-rooted vines and write $l_1 ,\dots, l_m$ for the orders of monodromy of the rooted vines. Let $l_{\max} = \max \{ l_1 ,\dots, l_m\}$, then we define the order of monodromy of the vineyard as the smallest $k$ that is a common multiple of $k_1, \dots , k_n$ and is larger than $l_{\max}$. We also call $k$ the common order of monodromy of all the vines.  

We can complete the vineyard in the same way as before, but with $k$ the common order of monodromy of all the vines, instead of an individual vine. We call the result the completed vineyard. 
\end{definition} 

We stress that there may be self-intersections on the diagonal of the completed vineyard. We note again here that if the vineyards are non-generic, the situation is significantly more complex, as vines can intersect~\cite{turner2023representing}; we will not consider the non-generic case any further in this work, but it may be of interest for future study.

\section{Omitted Proofs}
\label{appendix:proofs}

We now present more details on the technical proofs that were omitted in Section~\ref{sec:BraidsAndMorse}.

\subsection{Proof of Lemma~\ref{Lem:Step2}}

The proof of this lemma depends on one of the Morse theorems, see e.g. \cite[Theorem~3.1]{milnor1963morse}:
\begin{theorem}[Morse]\label{th:Morse} 
Suppose $f$ is a smooth real-valued function on $\M$, $a < b$, $f^{-1} [ a , b ]$ is compact, and there are no critical values between $a$ and $b$. Then $\M^a$ is diffeomorphic to $\M^{b}$, and $\M^{b}$ deformation retracts onto $\M^{a}$. 
\end{theorem}

\begin{proof}[of Lemma \ref{Lem:Step2}] 
We write $\beta (\theta)$ for a parametrization of $B$ according to the angle of the circle $C_h(0,R)$. We parametrize the circle as $s(\theta) =( R \sin \theta,R\cos \theta, 0 , \dots ,0 )$. Using this notation, we note that the gradient of $d(y, p)_\M$ is zero at $y=x$ if for its closest point on $B$, that is, $\beta(\theta)$, we have that 
$ \langle {\beta}',\beta- p  \rangle =0 $, where $\beta'$ denotes the derivative of $\beta$ with respect to $\theta$. This is equivalent to 
\[ \angle  {\beta}',\beta- p  = \pi/2.
\]
By assumption we have that 
\[  \angle  {\beta}' , s' \leq \theta_B.
\]
Moreover, because $| \beta(\theta) - s(\theta) | \leq \epsilon $, we also have 
\[ \sin (\angle \beta -p , s- p)   \leq  \frac{\epsilon} { |s -p| }, 
\]
as can be seen from Figure \ref{fig:Obvious}. 

\begin{figure}[h!]
    \centering
    \includegraphics[width=0.50\linewidth]{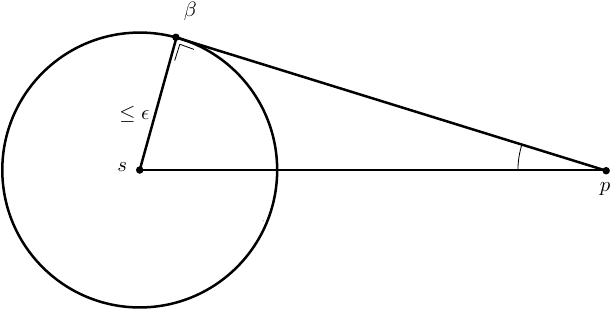}
    \caption{The angle estimate for $\angle \beta -p , s- p $. }
    \label{fig:Obvious}
\end{figure}

Using the triangle inequality of angles (or points on the sphere) we find that 
\begin{align} 
|\angle  ({\beta}',\beta- p)  -  \angle  ({s}',s- p) | \leq  \theta_B + \arcsin \frac{\epsilon} { |s -p| } . 
\label{eg:TriangleIneq1}
\end{align} 

Let us now write $d= | s- p|$, $p'$ for the closest point projection of $p$ on $C_h(0,R)$
and $\omega = \angle  ({s}',s- p')$. 
By reparameterization we can assume that $p'= s(0)$. With this assumption we see by the construction in Figure \ref{fig:AngleEst3} that  $\omega=\angle  ({s}',s- p') = \frac{\theta}{2}$. Moreover we have that $|p'- s | = 2 R \sin (\frac{\theta}{2}) $
\begin{figure}[h!]
    \centering
    \includegraphics[width=0.55\linewidth]{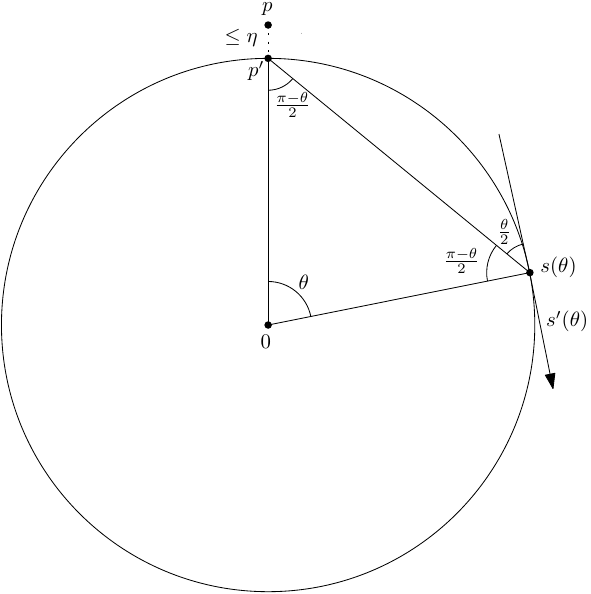}
    \caption{The construction for the angle $\omega= \angle  ({s}',s- p')$. }
    \label{fig:AngleEst3}
\end{figure}
By the same argument as that given in Figure \ref{fig:Obvious} we have  
\[  \sin \angle  (s-p',s- p)  \leq \frac{\eta}{|s-p'|},
\]
so that together with \eqref{eg:TriangleIneq1} we find that
\begin{align} 
\left |\angle  ({\beta}',\beta- p)  - \frac{\theta}{2}  \right|& \leq  \theta_B + \arcsin \frac{\epsilon} { |s -p| } +\arcsin \frac{\eta} { |s -p'| } 
\nonumber
\\
&\leq \theta_B + \arcsin \frac{\epsilon} { |s -p'| -\eta } +\arcsin \frac{\eta} { |s -p'| } 
\tag{by the triangle inequality} 
\\
&\leq \theta_B + \arcsin \frac{\epsilon} { 2 R \sin (\frac{\theta}{2})  -\eta } +\arcsin \frac{\eta} { 2 R \sin (\frac{\theta}{2})  } 
\nonumber
\end{align} 
This means that if 
\[ \frac{\theta}{2} + \theta_B + \arcsin \frac{\epsilon} { 2 R \sin (\frac{\theta}{2})  -\eta } +\arcsin \frac{\eta} { 2 R \sin (\frac{\theta}{2})  }  <\frac{\pi}{ 2} 
\] 
then  $\angle  ({\beta}',\beta- p) \neq \pi/2 $ and hence there are no critical points. The fact that there is no topological change follows from Theorem \ref{th:Morse}. 
\end{proof} 

\subsection{Proof of Lemma~\ref{Lem:Step3}}

\begin{proof}%
Because the conditions of Lemma \ref{Lem:Step2} are satisfied, we know that there are no Morse critical points unless $\theta \simeq 0$ or $\theta \simeq \pi$ so we focus on establishing that there is only one critical point per strand at $\theta \simeq 0$ or $\theta \simeq \pi$ for $B$ and or two per strand at $\theta \simeq 0$ or $\theta \simeq \pi$ in the case of $\M$.

The proof for the braid case is the difficult step, and we will see below that the statement for $\M$ follows immediately. The idea of the proof is the following: Because the closed braid is $\delta$-circular its parametrization $\beta$ satisfies $\left|\frac{\beta (t)}{R^2} +  \ddot{\beta} (t)  \right| <  \delta$. This means that $\beta$ is forced to turn inward towards the centre of $C_h(0,R)$, following the circle $C_h(0,R)$, where $0$ is the origin of Euclidean space. This in turn implies that $\beta$ cannot satisfy $\frac{d}{dt} | \beta - p|^2 (t) = 0 $ for two times $t$ that are relatively close. In other words, the curve cannot be tangent to some sphere (not necessarily of the same radius) centred at $p$. 

The way we establish that $\frac{d}{dt} | \beta - p|^2 (t) $ cannot be zero for two nearby values is by establishing that if $\frac{d}{dt} | \beta - p|^2 (t) =0 $  then  
\[
\frac{d}{dt} \left(\frac{d}{dt} | \beta - p|^2 \right) (t) = \frac{d^2}{dt^2} | \beta - p|^2  (t) 
\]
is large. 
If we write $\frac{\beta (t)}{R^2} +  \ddot{\beta} (t)= \Delta(t)$ and suppress $t$ from the notation, we see that 
\begin{align}
     \frac{d^2}{dt^2} | \beta - p|^2  (t) &= 2 \langle
     \ddot{\beta} , \beta -p \rangle  + 2 \langle  \dot{\beta} , \dot{ \beta} \rangle 
     \nonumber 
     \\ 
     &= 2 \langle
     \ddot{\beta} , \beta -p \rangle  + 2 \tag{because $|\dot{\beta} |=1$}
     \\
     &= 2 \left \langle - \frac{\beta}{R^2}  + \Delta , \beta -p \right \rangle  +2
     \nonumber
     \\
     &=   -2 \left \langle \frac{\beta}{R^2} , \beta \right \rangle  +2 \left \langle \Delta, \beta -p \right \rangle +2  +2 \left \langle \frac{\beta}{R^2} , p \right \rangle 
     \label{Eq:SecondDerDist}
\end{align}
We can now examine the first three terms in \eqref{Eq:SecondDerDist}: 
\begin{itemize} 
\item Because $B$ is $(\epsilon,R)$-embedded $R-\epsilon\leq  |\beta| \leq R+\epsilon$, so that 
\[ -2 \left (1+\frac{\epsilon}{R} \right )^2  \leq  -2 \left \langle \frac{\beta}{R^2} , \beta \right \rangle \leq -2 \left (1-\frac{\epsilon}{R} \right )^2,
\]
which, if $\epsilon \ll R $, simplifies to 
\[ -2 -\frac{6\epsilon}{R}  \leq  -2 \left \langle \frac{\beta}{R^2} , \beta \right \rangle \leq -2 +\frac{6\epsilon}{R} .
\]
\item Because $|\Delta(t)| = |\frac{\beta (t)}{R^2} +  \ddot{\beta} (t)| \leq \delta$, Cauchy-Schwarz yields that 
$ 2| \left \langle \Delta, \beta -p \right \rangle |\leq \delta |\beta- p| \leq \delta (R+\eta)$.  
\end{itemize}
This implies that these three terms are close to zero,  and hence %
\[
\frac{d^2}{dt^2} | \beta - p|^2  (t) \simeq  2 \left \langle \frac{\beta}{R^2} , p \right \rangle .
\]
Because $\left |2 \left \langle \frac{\beta}{R^2} , p \right \rangle \right|$ is lower bounded by
$\frac{1}{R^2} ( R- \epsilon) ( R- \eta) $
assuming that the angle between $\beta$ and $p$ is no more than $45$ degrees (or more than $135$), 
the first part of the result now follows if, the angle between $\beta$ and $p$ is no more than $45$ degrees (or more than $135$), $\epsilon \ll R$,  
\[ 
\frac{6\epsilon}{R} + \delta (R+\eta) \leq \frac{1}{R^2} ( R- \epsilon) ( R- \rho_{\max}) . 
\]

For the second part, note that there is a one-to-one correspondence between pairs of critical points of the distance function to a fixed point on an offset and the critical points of the distance function to the same fixed point on the curve $\beta$ itself. 
{Here minima of $\beta$ correspond to a pair of a minimum and a critical point of index $l$ on $\M$, while maxima correspond to a pair of a maximum and a critical point of index $1$.}
\end{proof} 

\subsection{Proof of Corollary~\ref{Lem:Step4}}

\begin{proof}
The only thing in this corollary that requires an extra argument on top of Lemma \ref{Lem:Step3}, is the correspondence of the critical points with the births and deaths respectively: because, by Remark \ref{Rem:ThetaEqZeroOrPi}  there are no Morse critical points unless $\theta \simeq 0$ or  $\theta \simeq \pi$, we know that $B(p, r) \cap B$ (respectively $B(p, r) \cap \M$) with $r \simeq R$ consists of $n$ topological line segments (topological cylinder segments $\mathbb{S}^{l+1} \times [0,1]$). While for $r > 2 R + \eta +\epsilon$, the set  $B(p, r) \cap B$ (respectively $B(p, r) \cap \M$) consists of a topological circle or knot (its offset respectively). See Figure \ref{fig:Step4}.  The only way this can be achieved with the number of critical points we found in Lemma \ref{Lem:Step3} is if the births and deaths occur as described in the statement of the corollary,  by a simple counting argument or the pigeonhole principle. 
We distinguish the two cases, namely $B$ and $\M$: For $B$ we have $n$ minima and $n$ maxima, and all $n$ minima are needed to create the $n$ topological line segments (which we know exist if $r \simeq R$).  Similarly, all $n$ maxima are needed to create the handle attachment that recovers the circle (which we know exists if $r > 2 R + \eta +\epsilon$). For $\M$,  we have  $n$ maxima and $n$ minima, as well as $n$ saddle points of index $l$ and $n$ of index $1$. We need $n$ minima and $n$ saddle points of index $l$ to create the $n$ topological cylinder segments $\mathbb{S}^{l+1} \times [0,1]$ (which we know exist if $r \simeq R$). To form $\mathbb{S}^{l+1} \times \mathbb{S}^1$ we need to connect these segments to each other, for which we need all the critical points of index $1$, so that we end up with $\mathbb{S}^{l+1} \times \mathbb{S}^1$ with $n$ punctures. We need all the maxima to fill all the punctures. 

\end{proof}

\end{document}